\documentclass[11pt,final]{article}
\usepackage{microtype}
\usepackage{xspace}
\usepackage{bbm}
\usepackage{graphicx}
\usepackage{amsmath,amssymb,amsthm}
\usepackage{tikz}

\usepackage{paralist}
\usepackage{bm}
\usepackage{xspace}
\usepackage{url}
\usepackage{fullpage, prettyref}
\usepackage{boxedminipage,wrapfig}
\usepackage{wrapfig}
\usepackage{ifthen}
\usepackage{color}
\usepackage{xcolor}
\usepackage{algpseudocode}
\usepackage{caption}
\usepackage{algorithm}
\usepackage{fullpage}
\usepackage{hyperref}
\usepackage[compact]{titlesec}
\usepackage{cleveref}
\usepackage{nicefrac}

\usepackage[bottom=1in,top=1in,left=1in,right=1in]{geometry}

\usepackage{framed}
\renewenvironment{leftbar}[1][\hsize]
{%
    \MakeFramed{\hsize#1\advance\hsize-\width\FrameRestore}%
}
{\endMakeFramed}

\usepackage{thmtools}
\usepackage{thm-restate}
\declaretheorem[numberwithin=section]{theorem}
\declaretheorem[sibling=theorem]{lemma}
\declaretheorem[sibling=theorem]{claim}

\newtheorem{definition}[theorem]{Definition}
\newtheorem{observation}[theorem]{Observation}
\crefname{claim}{claim}{claims}

\usepackage{framed}
\renewenvironment{leftbar}[1][\hsize]
{%
    \MakeFramed{\hsize#1\advance\hsize-\width\FrameRestore}%
}
{\endMakeFramed}

\usepackage[colorinlistoftodos,textsize=tiny,textwidth=2cm,color=red!25!white,obeyFinal]{todonotes}
\newcommand{\agnote}[1]{\todo[color=blue!25!white]{AG: #1}\xspace}
\newcommand{\elnote}[1]{\todo[color=green!25!white]{EL: #1}\xspace}

\newcommand{\alert}[1]{{\color{red}#1}\xspace}

\makeatletter
\algrenewcommand\algorithmiccomment[2][\normalsize]{{#1\hfill\(\triangleright\) \emph{#2}}}
\makeatother

\newcommand{\ignore}[1]{}

\newcommand{\eps}{\varepsilon}

\newcommand{\poly}{\mathrm{poly}}

\newcommand{\calP}{{\cal P}}

\newcommand{\calS}{{\cal S}}

\newcommand{\calC}{{\cal C}}

\newcommand{\br}{\boldsymbol{\mathrm{r}}}

\newcommand{\initOneLiners}{%
    \setlength{\itemsep}{0pt}
    \setlength{\parsep }{0pt}
    \setlength{\topsep }{0pt}
}
\newenvironment{OneLiners}[1][\ensuremath{\bullet}]
    {\begin{list}
        {#1}
        {\initOneLiners}}
    {\end{list}}

\newcommand{\Opt}{\ensuremath{\mathsf{Opt}}\xspace}
\newcommand{\sse}{\subseteq}

\usepackage{caption}

\DeclareCaptionFormat{algor}{%
  \hrulefill\par\offinterlineskip\vskip1pt%
    \textbf{#1#2}#3\offinterlineskip\hrulefill}
\DeclareCaptionStyle{algori}{singlelinecheck=off,format=algor,labelsep=space}
\newcommand{\kcut}{$k$\textsc{-Cut}\xspace} 
\newcommand{\mincut}{\mathsf{Mincut}} 
\newcommand{\mintwocut}{\mathsf{Mincut}} 
\newcommand{\minfourcut}{\text{\sf{Min-4-cut}}}
\newcommand{\opt}{\mathsf{Opt}} 
\newcommand{\Time}{\mathsf{Time}} 
\newcommand{\Laminarcut}[2]{\textsc{Laminar }#1\textsc{-cut}(#2)} 
\newcommand{\Laminarkcut}[1]{\textsc{Laminar }$#1$\textsc{-cut}} 
\newcommand{\PartialVC}[1]{\textsc{Partial }#1\textsc{-VC}} 

\newcommand{\Mincut}{\text{Mincut}} 
 
\newcommand{\Minfourcut}{\text{Min-4-cut}}
\newcommand{\Main}{\text{Main}} 
\newcommand{\Complete}{\text{Complete}} 
\newcommand{\Guess}{\text{Guess}}
\newcommand{\Laminar}{\text{Laminar}} 
\newcommand{\Record}{\text{Record}}

\newcommand{\kl}{\mathfrak{a}}
\newcommand{\kr}{\mathfrak{b}}

\newcommand{\Sr}{S^*_{\geq \mathfrak{b}}}
\newcommand{\nf}[2]{\nicefrac{#1}{#2}}

\newcommand{\As}[1]{\textbf{(A#1)}}

\newcommand{\pvc}{\textsc{Partial VC}\xspace}
\newcommand{\pvclong}{\textsc{Minimum Partial Vertex Cover}\xspace}

\newcommand{\inducedG}[1]{G[#1]}

\newcommand{\e}{\varepsilon}
\newcommand{\mT}{\mathcal T}

\newcommand{\mS}{\mathcal S}
\newcommand{\mB}{\mathcal B}
\newcommand{\solns}{\mathcal O}
\newcommand{\argmin}{\operatornamewithlimits{argmin}}

\newcommand{\desc}{\ensuremath{\mathrm{desc}}}
\newcommand{\anc}{\ensuremath{\mathrm{anc}}}
\newcommand{\children}{\ensuremath{\mathrm{children}}}
\newcommand{\subtree}{\ensuremath{\mathrm{subtree}}}
\newcommand{\saved}{\mathsf{Saved}} 

\newcommand{\minkut}{{\mu}}

\algnewcommand{\IIf}[1]{\State\algorithmicif\ #1\ \algorithmicthen}
\algnewcommand{\EndIIf}{\unskip\ \algorithmicend\ \algorithmicif}

\begin{document}

\title{{\bf An FPT Algorithm Beating 2-Approximation for $k$-Cut}}

\author{ Anupam Gupta\thanks{Supported in part by NSF awards
    CCF-1536002, CCF-1540541, and CCF-1617790.  This work was done in
    part when visiting the Simons Institute for the Theory of
    Computing. } \and Euiwoong Lee\thanks{Supported by NSF award CCF-1115525, Samsung scholarship, and Simons award for graduate students in TCS.}
  \and Jason Li\thanks{{\tt jmli@andrew.cmu.edu} }}

\date{Computer Science Department \\ Carnegie Mellon University \\ Pittsburgh, PA 15213.}

\thispagestyle{empty}
\maketitle
\begin{abstract}
  In the \kcut problem, we are given an edge-weighted graph $G$ and an
  integer $k$, and have to remove a set of edges with minimum total
  weight so that $G$ has at least $k$ connected components. Prior work
  on this problem gives, for all $h \in [2,k]$, a $(2-h/k)$-approximation
  algorithm for $k$-cut that runs in time $n^{O(h)}$. Hence to get a
  $(2 - \eps)$-approximation algorithm for some absolute constant
  $\eps$, the best runtime using prior techniques is
  $n^{O(k\eps)}$. Moreover, it was recently shown that getting a
  $(2 - \eps)$-approximation for general $k$ is NP-hard, assuming the
  Small Set Expansion Hypothesis.

  If we use the size of the cut as the parameter, an FPT algorithm to
  find the exact \kcut is known, but solving the \kcut problem exactly
  is $W[1]$-hard if we parameterize only by the natural parameter of
  $k$. An immediate question is: \emph{can we approximate \kcut better
    in FPT-time, using $k$ as the parameter?}

  We answer this question positively. We show that for some absolute
  constant $\eps > 0$, there exists a $(2 - \eps)$-approximation
  algorithm that runs in time $2^{O(k^6)} \cdot \widetilde{O} (n^4) $.  This is the first FPT
  algorithm that is parameterized only by $k$ and strictly improves the
  $2$-approximation.
\end{abstract}

\newpage

\setcounter{page}{1}

\section{Introduction}
\label{sec:introduction}

We consider the \kcut problem: given an edge-weighted graph $G =
(V,E,w)$ and an integer $k$, delete a minimum-weight set of edges so
that $G$ has at least $k$ connected components. This problem is a
natural generalization of the global min-cut problem, where the goal is
to break the graph into $k=2$ pieces.  
Somewhat surprisingly, the problem has poly-time algorithms for any
constant $k$: the current best result gives an $\tilde{O}(n^{2k})$-time
deterministic algorithm~\cite{Thorup08}. On the approximation algorithms
front, several $2$-approximation algorithms are
known~\cite{SV95, NR01, RS02}. Even a trade-off result is known: for any
$h \in [1,k]$, we can essentially get a $(2-\frac{h}{k})$-approximation in
$n^{O(h)}$ time~\cite{XCY11}. Note that to get $(2-\eps)$ for some
absolute constant $\eps > 0$, this algorithm takes time $n^{O(\eps k)}$,
which may be undesirable for large $k$. On the other hand, achieving a
$(2-\eps)$-approximation is NP-hard for general $k$, assuming the
Small Set Expansion Hypothesis (SSEH)~\cite{Manurangsi17}.

What about a better \emph{fine-grained result} when $k$ is small?
Ideally we would like a runtime of $f(k) \poly(n)$ so it scales better
as $k$ grows --- i.e., an FPT algorithm with parameter $k$.  Sadly, the
problem is $W[1]$-hard with this parameterization~\cite{DEFPR03}. (As an
aside, we know how to compute the optimal \kcut in time $f(|\Opt|) \cdot
n^{2}$~\cite{KT11, Chitnis}, where $|\Opt|$ denotes the cardinality of
the optimal \kcut.)
The natural question suggests itself: can we give a better approximation
algorithm that is FPT in the parameter $k$?

Concretely, the question we consider in this paper is: \emph{If we
  parameterize \kcut by $k$, can we get a $(2-\eps)$-approximation for
  some absolute constant $\eps > 0$ in FPT time---i.e., in time
  $f(k) \poly(n)$?}  (The hard instances which show $(2-\eps)$-hardness
assuming SSEH~\cite{Manurangsi17} have $k = \Omega(n)$, so such an FPT
result is not ruled out.) We answer the question positively.

\begin{theorem}[Main Theorem]
  \label{thm:kcut-main}
  There is an absolute constant $\eps > 0$ and an a
  $(2-\eps)$-approximation algorithm for the \kcut problem on general
  weighted graphs that runs in time $2^{O(k^6)} \cdot \tilde{O}(n^4)$.
\end{theorem}

Our current $\eps$ satisfies $\eps \geq 0.0003$ (see the calculations in
\S\ref{sec:conclusion}).  We hope that our result will serve as a
proof-of-concept that we can do better than the factor of~2 in FPT$(k)$
time, and eventually lead to a deeper understanding of the trade-offs
between approximation ratios and fixed-parameter tractability for the
\kcut problem. Indeed, our result combines ideas from approximation
algorithms and FPT, and shows that considering both settings
simultaneously can help bypass lower bounds in each individual setting,
namely the $W[1]$-hardness of an exact FPT algorithm and the
SSE-hardness of a polynomial-time $(2-\e)$-approximation.

To prove the theorem, we introduce two variants of \kcut.
\Laminarkcut{k} is a special case of \kcut where both the graph and the
optimal solution are promised to have special properties, and \pvclong
(\pvc) is a variant of \kcut where $k - 1$ components are required to be
singletons, which served as a hard instance for both the exact
$W[1]$-hardness and the $(2 - \eps)$-approximation SSE-hardness.  Our
algorithm consists of three main steps where each step is modular,
depends on the previous one: an FPT-AS for \pvc, an algorithm for
\Laminarkcut{k}, and a reduction from \kcut to \Laminarkcut{k}.  In the
following section, we give more intuition for our three steps.

\subsection{Our Techniques}
\label{sec:techniques}

For this section, fix an optimal $k$-cut
$\calS^* = \{ S^*_1, \dots, S^*_k\}$, such that
$w(\partial{S^*_1}) \leq \dots \leq w(\partial{S^*_k})$.  Let the
optimal cut value be
$\Opt := w(E(S^*_1, \dots, S^*_k)) = \sum_{i=1}^k w(\partial{S^*_i}) /
2$; here $E(A_1,\cdots, A_k)$ denotes the edges that go between
different sets in this partition.  The $(2 - 2/k)$-approximation iterative
greedy algorithm by Saran and Vazirani~\cite{SV95} repeatedly computes
the minimum cut in each connected component and takes the cheapest one
to increase the number of connected components by $1$. Its
generalization by Xiao et al.~\cite{XCY11} takes the minimum $h$-cut
instead of the minimum $2$-cut to achieve a $(2 - h/k)$-approximation in
time $n^{O(h)}$.

\subsubsection{Step I: \pvclong}  
\label{sec:overview-pvc}

The starting point for our algorithm is the $W[1]$-hardness result of
Downey et al.~\cite{DEFPR03}: the reduction from $k$-clique results in a
\kcut instance where the optimal solution consists of $k-1$ singletons
separated from the rest of the graph. Can we approximate such instances
well? Formally, the \pvc problem asks: given a edge-weighted graph, find
a set of $k-1$ vertices such that the total weight of edges hitting
these vertices is as small as possible?  Extending the result of
Marx~\cite{Marx07} for the maximization version, our first conceptual
step is an FPT-AS for this problem, i.e., an algorithm that given a
$\delta >0$, runs in time $f(k,\delta)\cdot \poly(n)$ and gives a
$(1+\delta)$-approximation to this problem.

\subsubsection{Step II: Laminar $k$-cut}  
\label{sec:overview-lam}

The instances which inspire our second idea are closely related to the
hard instances above. One instance on which the greedy algorithm of
Saran and Vazirani gives a approximation no better than $2$ for large
$k$ is this: take two cliques, one with $k$ vertices and unit edge
weights, the other with $k^2$ vertices and edge weights $1/(k+1)$, so
that the weighted degree of all vertices is the same. (Pick one vertex
from each clique and identify them to get a connected graph.) The
optimal solution is to delete all edges of the small clique, at cost
$\binom{k}{2}$. But if the greedy algorithm breaks ties poorly, it will
cut out $k-1$ vertices one-by-one from the larger clique, thereby
getting a cut cost of $\approx k^2$, which is twice as large. Again we
could use \pvc to approximate this instance well. But if we replace each
vertex of the above instance itself by a clique of high weight edges,
then picking out single vertices obviously does not work.  Moreover, one
can construct recursive and ``robust'' versions of such instances where
we need to search for the ``right'' (near-)$k$-clique to break
up. Indeed, these instances suggest the use of dynamic programming (DP),
but what structure should we use DP on?

One feature of such ``hard'' instances is that the optimal \kcut
$\calS^* = \{S_1^*, \ldots, S_k^*\}$ is composed of near-min-cuts in the
graph. Moreover, no two of these near-min-cuts cross each other. We now
define the \Laminarkcut{k} problem: find a \kcut on an instance where
none of the $(1+\eps)$-min-cuts of the graph cross each other, and where
each of the cut values $w(\partial{S_i^*})$ for $i = 1,\ldots, k-1$ are
at most $(1+\eps)$ times the min-cut. Because of this laminarity (i.e.,
non-crossing nature) of the near-min-cuts, we can represent the
near-min-cuts of the graph using a tree $\mT$, where the nodes of $G$
sit on nodes of the tree, and edges of $\mT$ represent the near-min-cuts
of $G$.  Rooting the tree appropriately, the problem reduces to
``marking'' $k-1$ incomparable tree nodes and take the near-min-cuts
given by their parent edges, so that the fewest edges in $G$ are
cut. Since all the cuts represented by $\mT$ are near-min-cuts and
almost of the same size, it suffices to mark $k-1$ incomparable nodes to
maximize the number of edges in $G$ both of whose endpoints lie below a
marked node. We call such edges \emph{saved} edges. \agnote{Any
  figures?} In order to get a $(2-\eps)$-approximation for \Laminarkcut{k}, 
it suffices to save $\approx \eps k \mincut$ weight of edges.

Note that if $\mT$ is a star with $n$ leaves and each vertex in $G$ maps
to a distinct leaf, this is precisely the \pvc problem, so we do not
hope to find the optimal solution (using dynamic programming, say).
Moreover, extending the FPT-AS for \pvc to this more general setting
does not seem directly possible, so we take a different approach. We
call a node an \emph{anchor} if has some $s$ children which when marked
would save $\approx \eps s \mincut$ weight. We take the following
``win-win'' approach: if there were $\Omega(k)$ anchors that were
incomparable, we could choose a suitable subset of $k$ of their children
to save $\approx \eps k \mincut$ weight. And if there were not, then all
these anchors must lie within a subtree of $\mT$ with at most $k$
leaves. We can then break this subtree into $2k$ paths and guess which
paths contain anchors which are parents of the optimal solution. For
each such guess we show how to use \pvc to solve the problem and save a
large weight of edges.  Finally how to identify these anchors? Indeed,
since all the mincuts are almost the same, finding an anchor again
involves solving the \pvc problem!

\subsubsection{Step III: Reducing \kcut to \Laminarkcut{k}}
\label{sec:overview-redn}


\begin{wrapfigure}{L}{0.38\textwidth}
  \centering
  \includegraphics[width=0.35\textwidth]{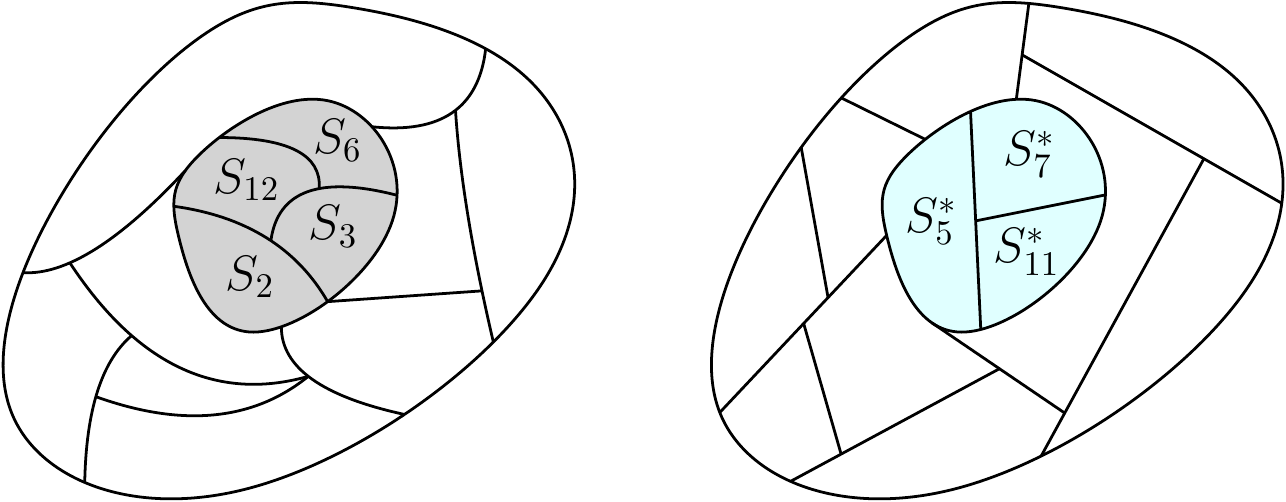}
  \caption{\label{fig:conform} The blue set on the right,
    formed by $S_5^* \cup S_7^* \cup S_{11}^*$, conforms to the
    algorithm's partition $\calS$ on the left.}
\end{wrapfigure}

We now reduce the general \kcut problem to \Laminarkcut{k}. This
reduction is again based on observations about the
graph structure in cases where the iterative greedy algorithms do not
get a $(2 - \varepsilon$)-approximation.  Suppose $\calS = \{ S_1, \dots,
S_{k'} \}$ be the connected components of $G$ at some point of an
iterative algorithm ($k' \leq k$).  For a subset $\emptyset \neq U
\subsetneq V$, we say that $U$ {\em conforms} to partition $\calS$ if
there exists a subset $J \subsetneq [k']$ of parts such that $U =
\cup_{j \in J} S_j$.  One simple but crucial observation is the
following: if there exists a subset $\emptyset \neq I \subsetneq [k]$ of
indices such that $\cup_{i\in I} S^*_i$ conforms to $\calS$
(i.e., $\cup_{i\in I} S^*_i = \cup_{j \in J} S_j$), we can
``guess'' $J$ to partition $V$ into the two parts $\cup_{i \in I} S^*_i$
and $\cup_{i \notin I} S^*_i$. Since the edges between these two parts
belong to the optimal cut and each of them is strictly smaller than $V$,
we can recursively work on each part without any loss.
 
Moreover, the number of choices for $J$ is at most
$2^{k'}$ and each guess produces one more connected component, so the
total running time can be bounded by $f(k)$ times the running time of
the rest of the algorithm, for some function $f(\cdot)$.  Therefore, we
can focus on the case where none of $\cup_{i \in I} S^*_i$ conforms to
the algorithm's partition $\calS$ at any point during the algorithm's
execution.

\begin{wrapfigure}{R}{0.38\textwidth}
  \centering
  \includegraphics[width=0.35\textwidth]{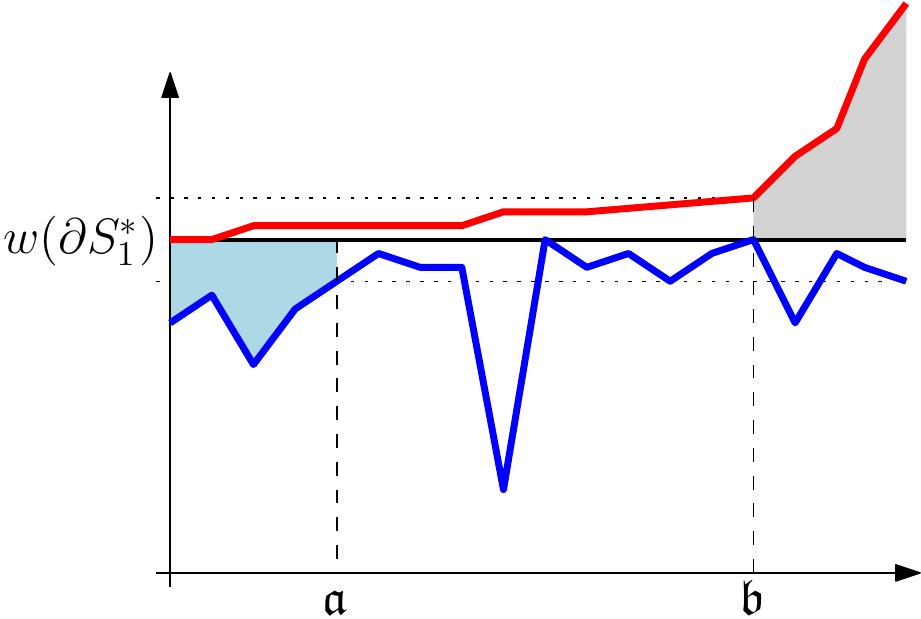}
  \caption{\label{fig:histo} The blue curve shows cut sizes
    for algorithm's cuts, red curve shows $w(\partial{S^*_i})$
    values. The blue area (and in fact all the area below
    $w(\partial{S_1^*})$ and above the algorithm's curve) makes the first
    inequality loose. The grey area (and in fact all the area above
    $w(\partial{S_1^*})$ and below OPT's curve) makes the second
    inequality loose.}
\end{wrapfigure}

Now consider the iterative min-cut algorithm of Saran and Vazirani, and
let $c_i$ be the cost of the min cut in the $i^{th}$ iteration ($1 \leq
i \leq k - 1$).  By our above assumption about non-conformity, none of
$\cup_{i \in I} S^*_i$, and in particular the subset $S^*_1$, conform to
the current components. This implies that deleting the remaining edges
in $\partial S^*_1$ is a valid cut that increases the number of
connected components by at least $1$, so $c_i \leq
w(\partial{S^*_1})$. Then we have the following chain of inequalities:
\[
\sum_{i=1}^{k-1} c_i \leq k \cdot w(\partial{S^*_1}) \leq \sum_{i=1}^k w(\partial{S^*_i}) = 2\Opt.
\]

If the iterative min-cut algorithm could not get a $(2 -
\varepsilon)$-approximation, the two inequalities above must be essentially
tight. Hence almost all our costs $c_i$ must be close to
$w(\partial{S^*_1})$ and almost all $w(\partial{S^*_i})$ must be close
to $w(\partial{S^*_1})$. 
Slightly more formally, let $\kl \in [k]$ be the smallest integer such
that 
$c_{\kl} \gtrsim w(\partial{S^*_1})$
---so that the first $\kl - 1$ cuts are ones where
we pay ``much'' less than $\partial{S^*_1}$ and make the first
inequality loose. And let $\kr \in [k]$ be the smallest number such that
$w(\partial{S^*_{\kr}}) \gtrsim w(\partial{S^*_1})$
--- so that the last
$k - \kr$ cuts in OPT are much larger than $\partial{S^*_1}$ and make
the second inequality loose. Then if the iterative min-cut algorithm is
no better than a $2$-approximation, we can imagine that $\kl = o(k)$ and
$\kr \geq k - o(k)$.
For simplicity, let us assume that $\kl = 1$ and $\kr = k$ here. 

Indeed, instead of just considering min-cuts, suppose we also consider
min-4-cuts, and take the one with better edges cut per number of new
components. The arguments of the previous paragraph still hold, so 
$\kl = 1$ implies that 
the best min-cuts and best min-4-way cuts (divided by 3) are roughly 
at least $w(\partial{S^*_1})$ in the original $G$. 
Since the min-cut is also at most $w(\partial{S^*_1})$, 
the weight of the min-cut is
roughly $w(\partial{S^*_1})$ and none of the near-min-cuts
cross (else we would get a good 4-way cut). I.e., the
near-min-cuts in the graph form a laminar family. 
Together with the fact that $\partial{S^*_1}, \dots, \partial{S^*_{k - 1}}$ are near-min-cuts (we assumed $\kr = k$),  
this is precisely an instance of \Laminarkcut{k}, which completes
the proof!

\paragraph{Roadmap.} After some related work and preliminaries, we first
present the details of the reduction from \kcut to \Laminarkcut{k} in
Section~\ref{sec:reduction}. Then in Section~\ref{sec:laminar} we give
the algorithm for \Laminarkcut{k} assuming an algorithm for \pvc. Finally
we give our FPT-AS for \pvc in Section~\ref{sec:partial-vc}.

\subsection{Other Related Work}
\label{sec:related}

The \kcut problem has been widely studied. Goldschmidt and Hochbaum gave
an $O(n^{(1/2- o(1))k^2})$-time algorithm~\cite{GH94}; they also showed
that the problem is NP-hard when $k$ is part of the input.  Karger and
Stein improved this to an $O(n^{(2-o(1))k})$-time randomized Monte-Carlo
algorithm using the idea of random edge-contractions~\cite{KS96}.
After Kamidoi et al.~\cite{KYN06} gave an $O(n^{4k + o(1)})$-time
deterministic algorithm based on divide-and-conquer, 
Thorup gave an $\tilde{O}(n^{2k})$-time deterministic algorithm based on
tree packings~\cite{Thorup08}. 
Small values of $k \in [2, 6]$ also have been separately studied~\cite{NI92, HO92, BG97, Karger00, NI00, NKI00, Levine00}.

On the approximation algorithms front, a $2(1-1/k)$-approximation was
given by Saran and Vazirani~\cite{SV95}.  Naor and Rabani~\cite{NR01},
and Ravi and Sinha~\cite{RS02} later gave $2$-approximation algorithms
using tree packing and network strength respectively.  Xiao et
al.~\cite{XCY11} completed the work of Kapoor~\cite{Kapoor96} and Zhao
et al.~\cite{ZNI01} to generalize Saran and Vazirani to essentially give
an $(2 - h/k)$-approximation in time $n^{O(h)}$.  Very recently,
Manurangsi~\cite{Manurangsi17} showed that for any $\eps > 0$, it is
NP-hard to achieve a $(2 - \eps)$-approximation algorithm in time
$\poly(n,k)$ assuming the Small Set Expansion Hypothesis.

\emph{FPT algorithms:} Kawarabayashi and Thorup give an
$f(\Opt) \cdot n^{2}$-time algorithm~\cite{KT11} for unweighted
graphs. Chitnis et al.~\cite{Chitnis} used a randomized color-coding
idea to give a better runtime, and to extend the algorithm to weighted
graphs. In both cases, the FPT algorithm is parameterized by the
cardinality of edges in the optimal \kcut, not by $k$. For a
comprehensive treatment of FPT algorithms, see the excellent
book~\cite{FPT-book}, and for a survey on approximation and FPT
algorithms, see~\cite{Marx07}.

\emph{Multiway Cut:} A problem very similar to \kcut is the
\textsc{Multiway Cut} problem, where we are given $k$ terminals and want
to disconnect the graph into at least $k$ pieces such that all terminals
lie in distinct components. However, this problem behaves quite
differently: it is NP-hard even for $k=3$ (and hence an $n^{f(k)}$
algorithm is ruled out); on the other hand several algorithms are
known to approximate it to factors much smaller than~$2$ (see,
e.g.,~\cite{BuchbinderSW17} and references therein). FPT algorithms
parameterized by the size of $\Opt$ are also known; see~\cite{CaoCF14}
for the best result currently known.

\section{Notation and Preliminaries}
\label{sec:prelims}

For a graph $G = (V,E)$, and a subset $S \sse V$, we use $G[S]$ to
denote the subgraph induced by the vertex set $S$. For a collection of
disjoint sets $S_1, S_2, \ldots, S_t$, let $E(S_1, \ldots, S_t)$ be the
set of edges with endpoints in some $S_i, S_j$ for $i \neq j$. Let
$\partial S = E(S, V \setminus S)$. We say two cuts $(A, V\setminus A)$
and $(B,V\setminus B)$ \emph{cross} if none of the four sets $A
\setminus B, B \setminus A, A \cap B$, and $V \setminus (A \cup B)$ is
empty.
$\mincut$ and $\minfourcut$ denote the weight of the min-2-cut
and the min-4-cut respectively. 
A cut $(A, V \setminus A)$ is called $(1 + \varepsilon)$-mincut if 
$w(A, V \setminus A) \leq (1 + \varepsilon) \mincut$. 
\begin{restatable}[\textsc{Laminar $k$-Cut}$(\eps_1)$]{definition}{LamDef}
  \label{def:laminarcut}
  The input is a graph $G = (V,E)$ with edge weights, and two parameters
  $k$ and $\e_1$, satisfying two promises: (i)~no two $(1+\e_1)$-mincuts
  cross each other, and (ii)~there exists a $k$-cut $\calS' = \{S_1',
  \ldots, S_k'\}$ in $G$ with $w(\partial(S_i')) \le (1+\e_1)\mincut(G)$
  for all $i \in [1,k-1]$. Find a $k$-cut with the total weight.
  The approximation ratio is defined as the ratio of 
  the weight of the returned cut to the weight of the \kcut $\calS'$
  (which can be possibly less than $1$).
\end{restatable}

\begin{definition}[\pvclong]
  \label{def:pvc}
  Given a graph $G = (V,E)$ with edge and vertex weights, and an integer
  $k$, find a vertex set $S \sse V$ with $|S| = k$ nodes, minimizing the
  weight of the edges hitting the set $S$ plus the weight of all
  vertices in $S$.
\end{definition}


\section{Reduction to $\Laminarcut{k}{\eps_1}$}
\label{sec:reduction}

In this section we give our reduction from \kcut to
$\Laminarcut{k}{\eps_1}$, showing that if we can get a better-than-2
approximation for the latter, we can beat the factor of two for the
general \kcut problem too. We assume the reader is familiar with the
overview in Section~\ref{sec:overview-redn}. Formally, the main theorem
is the following.

\begin{theorem}
  \label{thm:reduction1}
  Suppose there exists a $(2 - \eps_2)$-approximation algorithm for
  $\Laminarcut{k}{\eps_1}$ for some $\eps_1 \in (0, 1/4)$ and
  $\eps_2 \in (0, 1)$ that runs in time $f(k) \cdot g(n)$.  Then
  there exists a $(2 - \eps_3)$-approximation algorithm for \kcut
  that runs in time $2^{O(k^2  \log k)} \cdot f(k) \cdot (n^4 \log^3 n + g(n))$ for some constant
  $\eps_3 > 0$. 
\end{theorem}

\begin{algorithm}
        \caption{$\Main(G  = (V, E, w), k)$}
        \label{alg:main}
        \begin{algorithmic}[1]
                \State $k' = 1$, $S_1 \gets V$
                \While {$k' < k$ } 
                \For {$\br \in [k]^{k'}$ } \label{line:start-of-check} \Comment {Further partition each $S_i$ into $r_i$ components by Laminar}
                \State $|\br| \gets \sum_{j=1}^{k'} r_j$; $\{ C_1, \dots,
                C_{|\br|} \} \gets \cup_{i \in [k']}
                \Laminar(\inducedG{S_i}, r_i)$. 
		      \If{$|\br| \geq k$}  $C_k \gets C_k \cup \dots \cup C_{|\br|}$ 
                \Else \State $\{C_1, \dots, C_{k}\} \gets  \Complete(G, k, C_1, \dots, C_{|\br|})$    \label{line:complete}
		      \EndIf
                \State \Record($\Guess(\{C_1, \dots, C_k \})$) 
                \EndFor \label{line:end-of-check}
                \State {} \Comment {Split some $S_i$ by a mincut or a min-4-cut}
                \If{$k' > k - 3$ or $\min_{i \in [k']} \mintwocut (\inducedG{S_{i}}) \leq
                \min_{i \in [k']} \minfourcut (\inducedG{S_{i}}) / 3$} \label{line:start-extend}
                \State $i \gets \min_{i} \mintwocut (\inducedG{S_{i}})$; $\{ T_1, T_2 \} \gets \Mincut(\inducedG{S_{i}})$
                \State $S_i \gets T_1$; $S_{k' + 1} \gets T_2$; $c_{k'} \gets \mintwocut (\inducedG{S_{i}})$; $k' \gets k' + 1$
                \Else
                \State $i \gets \arg\min_{i} \minfourcut (\inducedG{S_{i}})$; $\{ T_1, \dots, T_4 \} \gets \Minfourcut(\inducedG{S_{i}})$;
                $S_{i} \gets T_1$
                \State $S_{k' + 1}, S_{k' + 2}, S_{k' + 3} \gets T_2, T_3, T_4$; 
                $c_{k'}, c_{k' + 1}, c_{k' + 2} \gets \minfourcut (\inducedG{S_{i}}) / 3$;
                $k' \gets k' + 3$
                \EndIf \label{line:end-extend}
                \EndWhile 
                \State let $\calS = \{S_1, \ldots, S_k\}$ be the final reference $k$-partition.
                \State \Record($\Guess(G, k, \calS)$) \label{line:lastupdate}

                \State Return the best recorded $k$-partition. 
        \end{algorithmic}
\end{algorithm}

\begin{algorithm}
\caption{$\Complete(G = (V, E, w), k, \calC = \{C_1, \ldots, C_\ell\})$}
\label{alg:complete}
\begin{algorithmic}[1]
\While {$\ell < k$}
\State $i \gets \min_{i \in [\ell]} \mincut(\inducedG{C_i})$; $T_1, T_2 \gets \Mincut(\inducedG{C_i})$
\State $C_i \gets T_1$; $C_{\ell + 1} \gets T_2$; $\ell \gets \ell + 1$
\EndWhile
\State Return $\calC := \{C_1, \dots, C_k\}$.
\end{algorithmic}
\end{algorithm}

\begin{algorithm}
        \caption{$\Guess(G  = (V, E, w), k, \calC = \{C_1, \dots, C_k\})$}
        \label{alg:guess}
        \begin{algorithmic}[1]
                \State \Record($C_1, \dots, C_k$)  \Comment{Returned
                  partition no worse than starting partition}
                \For {$\emptyset \neq J \subsetneq [k]$ } 
                \For {$k' = 1, 2, \dots, k - 1$ }
                \State $L \gets \cup_{j \in J} C_j$; $R \gets V
                \setminus L$ \Comment {Divide $S_i$ into two groups,
                  take union of each group}
                \State $D_1, \dots, D_{k'} \gets  \Main(\inducedG{L},
                k')$ \Comment{and recurse}
                \State $D_{k'+1}, \dots, D_k \gets  \Main(\inducedG{R}, k - k')$
                \State \Record($D_1, \dots, D_k$)
                \EndFor
                \EndFor
                \State Return the best recorded $k$-partition among all these guesses. 
        \end{algorithmic}
\end{algorithm}

The main algorithm is shown in Algorithm~\ref{alg:main} (``\Main''). 
It maintains a ``reference'' partition $\calS$, which is initially the
trivial partition where all vertices are in the same part. At each
point, it guesses how many pieces each part $S_i$ of this reference partition
$\calS$ should be split into using the ``Laminar'' procedure, and then
extends this to a $k$-cut using greedy cuts if necessary
(Lines~\ref{line:start-of-check}--\ref{line:end-of-check}).
It then
extends the reference partition by either taking the best min-cut or the
best min-4-cut among all the parts
(Lines~\ref{line:start-extend}--\ref{line:end-extend}). 

Every time it
has a $k$-partition, it guesses (using ``Guess'') if the union of some
of the parts equals some part of the optimal partition, and uses that to
try get a better partition. 
If one of the guesses is right, we strictly increase the number of
connected components by deleting edges in the optimal $k$-cut, so we can
recursively solve the two smaller parts.  If none of our guesses was
right during the algorithm, our analysis in
Section~\ref{subsec:approx_factor} shows that there exist values of $k',
\br$ such that $\calC = \{C_1, \dots, C_k\}$ in
Line~\ref{line:complete}, obtained from the reference partition $\calS =
\{ S_1, \dots, S_{k'} \}$ by running Laminar($G[S_{i}], r_i$) for each $i \in [k']$ 
and using Complete if necessary to get $k$ components, beats the $2$-approximation. Finally, a couple words about
each of the subroutines.
\begin{itemize}
\item Mincut$(G = (V, E, w))$ (resp.\ Min-4-cut$(G)$) 
returns the minimum $2$-cut (resp.\ $4$-cut) as a partition of $V$
into $2$ (resp. $4$) subsets. 
\item The subroutine ``Laminar'' returns a $(2-\eps_2)$-approximation
  for \Laminarcut{$k$}{$\eps_1$}, using the algorithm from
  Theorem~\ref{thm:laminar}. Recall the definition of the problem in Definition~\ref{def:laminarcut}. 

\item The operation ``\Record($\calP$)'' in \Guess\ and \Main\ takes a
  $k$-partition $\calP$ and compares the weight of edges crossing this
  partition to the least-weight $k$-partition recorded thus far (within
  the current recursive call). If the current partition has less weight,
  it updates the best partition accordingly. 

\item Algorithm~\ref{alg:complete}(``\Complete'') is a simple algorithm
  that given an $\ell$-partition $\calP$ for some $\ell \leq k$, outputs a
  $k$-partition by iteratively taking the mincut in the current graph.

\item Algorithm~\ref{alg:guess}(``\Guess''), when given an
  $\ell$-partition $\calP$ ``guesses'' if the vertices belonging to some
  parts of this partition $\{ S_j \}_{j \in J}$ coincide with the union
  of some $k'$ parts of the optimal partition.  If so, we have made
  tangible progress: it recursively finds a small $k'$-cut in the graph
  induced by $\cup_{j \in J} S_j$, and a small $k-k'$ cut in the
  remaining graph. It returns the best of all these guesses.
\end{itemize}

\subsection{The Approximation Factor}
\label{subsec:approx_factor}


\begin{lemma}[Approximation Factor]
  \label{lem:apx-main}
  $\Main(G, k)$ achieves a $(2 - \eps_3)$ approximation for some
  $\eps_3 > 0$ that depends on $\eps_1, \eps_2$ in
  Theorem~\ref{thm:reduction1}.
\end{lemma}

\begin{proof}
  We prove the lemma by induction on $k$.  The value of $\eps_3$ will be
  determined later.  The base case $k = 1$ is trivial.  Fix some value
  of $k$, and a graph $G$.  Let $\calS = \{S_1, \dots, S_k\}$ be the
  final reference partition generated by the execution of $\Main(G, k)$,
  and let $c_1, \dots, c_{k - 1}$ be the values associated with it. From the
  definition of the $c_i$'s in Procedure~\Main, $\sum_{i=1}^{k - 1} c_i =
  w(E(S_1, \dots, S_k))$. The $k$-partition returned by $\Main(G, k)$ is
  no worse than this partition $\calS$ (because of the update on
  line~\ref{line:lastupdate}), and hence has cost at most $\sum_{i=1}^{k-1}
  c_i = w(E(S_1, \dots, S_k))$.  Let us fix an optimal $k$-cut $\calS^*
  = \{ S^*_1, \dots, S^*_k\}$, and let $w(\partial{S^*_1}) \leq \dots
  \leq w(\partial{S^*_k})$.  Let $\Opt := w(E(S^*_1, \dots, S^*_k)) =
  \sum_{i=1}^k w(\partial{S^*_i}) / 2$.

  \begin{definition}[Conformity]
    \label{def:conform}
    For a subset $\emptyset \neq U \subsetneq V$, we say that $U$ {\em
      conforms} to partition $\calS$ if there exists a subset $J
    \subsetneq [k]$ of parts such that $U = \cup_{j \in J} S_j$. (See Figure~\ref{fig:conform}.)
  \end{definition}
  The following claim shows that if there exists a subset $\emptyset
  \neq I \subsetneq [k]$ of indices such that $\cup_{i\in I} S^*_i$
  conforms to $\calS$, the induction hypothesis guarantees a $(2 -
  \eps_3)$-approximation.
  \begin{claim}
    Suppose there exists a subset $\emptyset \neq I \subsetneq [k]$ such
    that $\cup_{i \in I} S_i^*$ conforms to $\calS$. Then $\Main(G, k)$
    achieves a $(2 - \eps_3)$-approximation.
    \label{claim:good}
  \end{claim}
  \begin{proof}
    Since $S^*_I := \cup_{i \in I} S^*_i$ conforms to $\calS$, during
    the run of $\Guess(G, k, \calS)$ it will record the $k$-partition
    $(\Main(\inducedG{S^*_I}, |I|), \Main(\inducedG{V \setminus S^*_I},
    k - |I|) )$, and hence finally output a $k$-partition which cuts no
    more edges than this starting partition.  By the induction
    hypothesis, $\Main(\inducedG{S^*_I}, |I|)$ gives a $|I|$-cut of
    $\inducedG{S^*_I}$ whose cost is at most $(2 - \eps_3)$ times
    $w(E(S^*_i)_{i \in I})$, and $\Main(\inducedG{V \setminus S^*_I}, k
    - |I|)$ outputs a $(k - |I|)$-cut of $\inducedG{V\setminus S^*_I}$
    of cost at most $(2 - \eps_3)$ times $w(E(S^*_i)_{i \notin
      I})$. Thus, the value of the best $k$-partition returned by
    $\Main(G, k)$ is at most
    \begin{align*}
      & w(E(S^*_I, V \setminus S^*_I)) + (2 - \eps_3) \left(
        w(E(S^*_i)_{i \in I}) +  w(E(S^*_i)_{i \notin I}) \right) \\
      \leq  & \ (2 - \eps_3) w(E(S^*_1, \dots, S^*_k)) = (2 -
      \eps_3) \opt. \qedhere 
    \end{align*}
  \end{proof}
  Therefore, to prove Lemma~\ref{lem:apx-main}, it suffices to assume
  that no collection of parts in \Opt conforms to our partition at any
  point in the algorithm. I.e.,
  \begin{leftbar}
    \As{1}: for every subset $\emptyset \neq I \subsetneq [k]$, $\cup_{i
      \in I} S^*_i$ does not conform to $\calS = \{ S_1, \dots, S_k\}$.
  \end{leftbar}

  Next, we study how $\Opt$ is related to $w(\partial{S_1^*})$.  Note
  that $\Opt \geq (k/2) \cdot w(\partial{S_1^*})$.  The next claim shows
  that we can strictly improve the $2$-approximation if $\Opt$ is even
  slightly bigger than that.

  \begin{claim}
    For every $i = 1, \dots, k-1$, $c_i \leq w(\partial{S^*_1})$. Moreover,
    if $\Opt \geq (k - 1)w(\partial{S_1^*}) / (2 - \eps_3)$, $\Main(G,
    k)$ achieves a $(2 - \eps_3)$-approximation.
    \label{claim:notgood}
  \end{claim}
  \begin{proof}
    Consider the beginning of an arbitrary iteration of the while loop
    of $\Main(G, k)$. Let $k'$ and $\calS' = \{ S_1, \dots, S_{k'} \}$
    be the values at that iteration. By \As{1}, set $S_1^*$ does not
    conform to $\calS'$ (because $\calS'$ only gets subdivided as the
    algorithm proceeds, and $S_1^*$ does not conform to the final partition
    $\calS$). So there exists some $i \in [k']$ such that $S_i$
    intersects both $S^*_1$ and $V \setminus S^*_1$. If we consider
    $\inducedG{S_i}$ and its mincut,
    \[
    \mincut(\inducedG{S_i}) \leq 
    w(E(S_i \cap S^*_1, S_i \setminus S^*_1))
    \leq w(\partial{S^*_1}).
    \]
    Now the new $c_j$ values created in this iteration of the while loop
    are at most the smallest mincut value, so we have that each $c_j
    \leq w(\partial{S_1^*})$.  Therefore,
    \[
    w(E(S_1, \dots, S_k)) = \sum_{i=1}^{k - 1} c_i \leq (k - 1)\cdot
    w(\partial{S_1^*}), 
    \]
    and $\Main(G, k)$ achieves a $(2 - \eps_3)$-approximation if $(k
    - 1) w(\partial{S^*_1}) \leq (2 - \eps_3) \Opt$.
  \end{proof}
  Consequently, it suffices to additionally assume that $\Opt$ is
  close to $(\nicefrac{k}{2}) \, w(\partial{S^*_1})$. Formally,
  \begin{leftbar}
    \As{2}: $ \Opt < w(\partial{S^*_1}) \cdot \frac{k - 1}{2 - \eps_3} $.
  \end{leftbar}
  Recall that $\eps_1, \eps_2 > 0$ are the parameters such that
  there is a $(2 - \eps_2)$-approximation algorithm for
  $\Laminarcut{k}{\eps_1}$.  Let $\kl \in [k]$ be the smallest
  integer such that $c_{\kl} > w(\partial{S^*_1}) (1 -
  \nicefrac{\eps_1}{3})$ (set $\kl = k$ if there is no such integer).
  (See Figure~\ref{fig:histo}.)
  In other words, $\kl$ is the value of $k'$ in the while loop of
  $\Main(G, k)$ when both $\min_i \mincut(\inducedG{S_i})$ and $\min_i
  \minfourcut(\inducedG{S_i}) / 3$ are bigger than $w(\partial{S^*_1}) (1 -
  \nicefrac{\eps_1}{3})$ for the first time.  Let $\eps_4 > 0$ be a constant
   satisfying
   \begin{equation}
     \label{eq:para_1}
     (2/3) \cdot \eps_1 \eps_4 \geq \eps_3.
   \end{equation}
   The next claim shows that we are done if $\kl$ is large. 
   \begin{claim}
     If $\kl \geq \eps_4 k$, $\Main(G, k)$ achieves a $(2 -
     \eps_3)$-approximation.
    \label{clm:left_tail}
   \end{claim}
   \begin{proof}
     If $\kl \geq \eps_4 k$, we have
     \begin{align*}
       \sum_{i=1}^{k-1} c_i &\leq 
       (\kl - 1) (1 - \nicefrac{\eps_1}{3})\cdot w(\partial{S^*_1}) + 
       (k - \kl)\cdot w(\partial{S^*_1})  \\
       & \leq k \cdot w(\partial{S^*_1})\cdot (1 - \nicefrac{\eps_1 \eps_4}{3})
       \leq (2 - (\nicefrac23) \eps_1 \eps_4) \Opt \leq (2 - \eps_3)
       \Opt.  \qedhere 
     \end{align*}
   \end{proof}
   Thus, we can assume that our algorithm finds very few cuts appreciably
   smaller than $w(\partial{S^*_1})$.
   \begin{leftbar}
     \As{3}: $\kl < \eps_4 k$.
   \end{leftbar}
   Let $\kr \in [k]$ be the smallest number such that
   $w(\partial{S^*_{\kr}}) > w(\partial{S^*_1}) (1 +
   \nicefrac{\eps_1}{3})$; let it be $k$ if there is no such
   number. (Again, see Figure~\ref{fig:histo}.)  Observe that $\kl$ is
   defined based on our algorithm, whereas $\kr$ is defined based on the
   optimal solution.  Let $\eps_5 > 0$ be a constant satisfying:
   \begin{equation}
     \label{eq:para_2}
     \frac{1}{2 - \eps_3} \leq \frac{1 + \nf{\eps_1 \eps_5}{3}}{2}
     ~~~\Leftrightarrow~~~ (1 + \nf{\eps_1 \eps_5}{3})(2 -
     \eps_3) \geq 2.  
   \end{equation}
   The next claim shows that $\kr$ should be close to $k$. 
   \begin{claim}
     $\kr \geq (1 - \eps_5)k$. 
   \end{claim}
   \begin{proof}
     Suppose that $\kr <(1 - \eps_5) k$. We have
     \begin{align*}
       & \ \frac{ k \cdot w(\partial{S^*_1}) }{2 - \eps_3} \stackrel{\As{2}}{>} \Opt = \frac{1}{2} \sum_{i = 1}^k w(\partial{S^*_i}) \\
       \geq& \ \frac{w(\partial{S^*_1})}{2} \left( (1 - \eps_5)k +
         \eps_5 k (1 + \eps_1 / 3) \right)
       = 
       \frac{k\cdot w( \partial{S^*_1} )}{2} \left( 1 + \nf{\eps_1 \eps_5}{3} \right),
     \end{align*}
     which contradicts~\eqref{eq:para_2}. 
   \end{proof}
   Therefore, we can also assume that very few cuts in \Opt 
   are appreciably larger than $w( \partial{S^*_1})$.
   \begin{leftbar}
     \As{4}: $\kr \geq (1 - \eps_5) k$. 
   \end{leftbar}

   \medskip\textbf{Constructing an Instance of Laminar Cut:} In order to
   construct the instance for the problem, let $\Sr = \cup_{i = \kr}^{k}
   S^*_i$ be the union of these last few components from $\calS^*$ which
   have ``large'' boundary.  Consider the iteration of the while loop
   when $k' = \kl$ and consider $S_1, \dots, S_{\kl}$ in that
   iteration. By its definition, $c_{\kl} > w(\partial{S_1^*}) (1 -
   \nicefrac{\eps_1}{3})$. Hence
   \begin{gather}
     \min_i \mincut(\inducedG{S_i}) > w(\partial{S_1^*}) (1 -
     \nicefrac{\eps_1}{3}),  
     \label{eq:stop1} \\
     \min_i \minfourcut(G[S_i]) > 3 w(\partial{S_1^*}) (1 -
     \nicefrac{\eps_1}{3}). 
     \label{eq:stop2}
   \end{gather}
   In particular, \eqref{eq:stop2} implies that no two near-min-cuts
   cross, since two crossing near-min-cuts will result in a $4$-cut of
   weight roughly at most $2 w(\partial{S_1^*})$.
   However, we are not yet done, since we need to factor out the effects
   of the $\kl - 1$ ``small'' cuts found by our algorithm. For this, we need
   one further idea.

   Let $\br = (r_1, r_2, \ldots, r_{\kl}) \in [k]^{\kl}$ be such that
   $r_i$ is the number of sets
   $S^*_1, \dots, S^*_{\kr - 1}, S^*_{\geq \kr}$ that intersect with
   $S_i$, and let $|\br| := \sum_{i = 1}^{\kl} r_i$.  If we consider the
   bipartite graph where the left vertices are the algorithm's
   components $S_1, \dots, S_{\kl}$, the right vertices are
   $S^*_1, \dots, S^*_{\kr - 1}, S^*_{\geq \kr}$, and two sets have an
   edge if they intersect, then $|\br|$ is the number of edges.  Since
   there is no isolated vertex and the graph is connected (otherwise
   there would exist $\emptyset \neq I \subsetneq [k']$ and
   $\emptyset \neq J \subsetneq [k]$ with
   $\cup_{i \in I}S_i = \cup_{j \in J} S^*_j$ contradicting~\As{1}),
   the number of edges is $|\br| \geq \kl + \kr - 1$.

   \begin{claim}
     \label{clm:promises}
     For each $i$ with $r_i \geq 2$, 
     the graph $\inducedG{S_i}$ satisfies the two promises of the problem
     $\Laminarcut{r_i}{\eps_1}$.
   \end{claim}
   \begin{proof}
     Fix $i$ with $r_i \geq 2$. Let
     $J := \{ j \in [\kr - 1] \mid S_i \cap S^*_j \neq \emptyset \}$ be
     the sets $S_j^*$ among the first $\kr - 1$ sets in the optimal
     partition that intersect $S_i$.  Since $|J| \geq r_i - 1$ and $r_i \geq 2$, $|J| \geq 1$.
     Note that
     $(1 - \nf{\eps_1}3)\cdot w(\partial{S^*_1}) <
     \mincut(\inducedG{S_i})$ by~\eqref{eq:stop1}.
     For every $j \in J$,
     \[
     \mincut(\inducedG{S_i}) \leq 
     w(E(S_i \cap S^*_j, S_i \setminus S^*_j)) \leq w(\partial{S^*_j})
     \leq (1 + \eps_1 / 3)\; w(\partial{S^*_1}) \leq (1 + \eps_1)
     \;\mincut(\inducedG{S_i}).
     \]
     The first and second inequality hold since both parts
     $S_i \cap S^*_j$ and $S_i \setminus S^*_j$ are nonempty, and hence
     deleting all the edges in $\partial{S_j^*}$ would separate
     $G[S_i]$.
     The third inequality is by the choice of $\kr$, and the last
     inequality uses~(\ref{eq:stop1}) and the fact that
     $(1 + \nicefrac{\eps_1}{3}) \leq (1 + \eps_1)(1 -
     \nicefrac{\eps_1}{3})$ when $\eps_1 < 1/4$.  

     This implies that in $\inducedG{S_i}$, for every $j \in J$,
     $(S_i \cap S^*_j, S_i \setminus S^*_j)$ is a $(1 + \eps_1)$-mincut.
     Furthermore, in $\inducedG{S_i}$, no two $(1+\eps_1)$-mincuts cross
     because it will result a 4-cut of cost at most
     \[
     2(1+\eps_1)\; \mincut(\inducedG{S_i}) \leq 2(1+\eps_1) (1 + \eps_1 / 3)
     \; w(\partial{S_1^*}),
     \]
     contradicting~\eqref{eq:stop2}. (Note that $2(1 +
     \eps_1) (1 + \nicefrac{\eps_1}{3}) \leq 3(1 - \nicefrac{\eps_1}{3})$ when
     $\eps_1 < 1/4$.) Hence, in $\inducedG{S_i}$, the two promises for
     $\Laminarcut{r_i}{\eps_1}$ are satisfied.
   \end{proof}
   Our algorithm $\Main(G, k$) runs $\Laminar(\inducedG{S_i}, r_i)$ for
   each $i \in [\kl]$ when it sets $k' = \kl$ and the vector $\br$ as
   defined above.  As in the algorithm, let
   $\calC = \{C_1, \dots, C_{k}\}$ be the partition obtained in
   Line~\ref{line:complete}.  In other words, to obtain the $k$ sets
   $C_1, \dots, C_k$ from the set $V$, we take the reference partition
   $S_1, \dots, S_{\kl}$ and further partition these sets using Laminar
   to get $|\br|$ parts $C_1, \dots, C_{|\br|}$. If $|\br| \geq k$, we
   can merge the last $|\br| - k + 1$ parts to get exactly $k$ parts if
   we want (but we will not take any edge savings into account in this
   calculation). If $|\br| < k$, we get $k - |\br|$ more parts using
   the Complete procedure.  

   The total cost of this solution $\calC$ is $w(E(C_1, \dots, C_k))$,
   which is $\sum_{j=1}^{\kl - 1} c_j \leq (\kl - 1) w(\partial{S^*_1})$
   plus the cost of $\Laminar(\inducedG{S_i}, r_i)$ for all
   $i \in [\kl]$ and the cost of $\Complete$.  Since
   Claim~\ref{clm:promises} considers the partition of each
   $\inducedG{S_i}$ obtained by cutting edges belonging to the optimal
   $k$-partition, 
   the sum of the cost of the $r_i$-partition we compare to in 
   each \Laminarkcut{r_i} is exactly $\Opt$.
   Hence the cost of the
   solution given by $\Laminar(\inducedG{S_i}, r_i)$ summed over
   $i \in [\kl]$ is bounded by $(2 - \eps_2) \Opt$, by the approximation
   assumption in Theorem~\ref{thm:reduction1}. 

   If $\cup_{i \in I} S^*_i$ for some $\emptyset \neq I \subsetneq [k]$
   conforms to $\calC$, then since \Main\ also records $\Guess(\calC)$,
   the proof of Claim~\ref{claim:good} guarantees that $\Main(G, k)$
   gives a $(2 - \eps_3)$ approximation using the induction hypothesis.
   Otherwise, $S^*_1$ does not conform to $\calC$, so the arguments used
   in the proof of Claim~\ref{claim:notgood} show that the cost of
   $\Complete$ is at most $(k - |\br|)\, w(\partial{S^*_1})$ if
   $|\br| \leq k$, and $0$ otherwise.  Since $|\br| \geq \kl + \kr - 1$,
   the total cost $w(E(C_1, \dots, C_k))$ is then bounded
   by
   \begin{align*}
     & (\kl - 1 )  w(\partial{S^*_1}) + (2 - \eps_2) \Opt + 
     (k - \kl - \kr + 1)  w(\partial{S^*_1}) \\
     = & \ (2 - \eps_2) \Opt + (k - \kr)  w(\partial{S^*_1}) \\
     \leq & \ (2 - \eps_2) \Opt + \eps_5  k \cdot
     w(\partial{S^*_1})  \tag{by \As{4}}\\ 
     \leq & \ (2 - \eps_2 + 2 \eps_5) \Opt. 
   \end{align*}
   Therefore, if
   \begin{equation}
     \eps_3 \leq \eps_2 -2 \eps_5,
     \label{eq:para_4}
   \end{equation}
   then $\Main(G, k)$ gives a $(2 - \eps_3)$ approximation in every
   possible case.  We set $\eps_3, \eps_4, \eps_5 > 0$ so that they
   satisfy the three conditions~\eqref{eq:para_1}, \eqref{eq:para_2},
   and~\eqref{eq:para_4}, namely,
   \[
   (2/3) \cdot \eps_1 \eps_4 \geq \eps_3, \quad (1 +
   \eps_1 \eps_5 / 3)(2 - \eps_3) \geq 2, 
   \quad \eps_3 \leq \eps_2 - 2 \eps_5.
   \]
   (For instance, setting $\eps_4 = \eps_5 = \min(\eps_1,
   \eps_2) / 3$ and $\eps_3 = \eps_4^2$ works.)
\end{proof}

\subsection{Running Time}

We prove that this algorithm also runs in FPT time, finishing the proof
of Theorem~\ref{thm:reduction1}.
\begin{lemma}
Suppose that $\Laminar(G, k)$ runs in time $f(k) \cdot g(n)$. 
Then Main$(G, k)$ runs in time $2^{O(k^2 \log k)} \cdot f(k) \cdot (g(n) + n^4 \log^3 n)$.
\end{lemma}
\begin{proof}
Let $\Time(\text{P})$ denote the running time of a procedure \text{P}. 
Here each procedure is only parameterized by the number of sets it outputs (e.g., $\Main(k), \Guess(k), \Complete(k), \Laminar(k)$). 
We use the fact that the global min-cut can be computed in time $O(n^2 \log^3 n)$~\cite{KS96} and the min-$4$-cut can be computed in $O(n^4 \log^3 n)$~\cite{Levine00}.
First, $\Time(\Complete(k)) = O(kn^2 \log^3 n)$. For $\Guess$ and $\Main$, 
\[
\Time(\Guess(k)) \leq k \cdot 2^{k + 1} \cdot ( \Time(\Main(k - 1)) + O(n) ),
\]
and 
\begin{align*}
\Time(\Main(k)) & \leq k^k \cdot (\Time(\Laminar(k)) + 
 \Time(\Guess(k)) + \Time(\Complete(k))) + O(k n^4 \log^3 n) \\
& \leq 2^{O(k \log k)} \cdot f(k) \cdot (g(n) + O(n^4 \log^3 n)) + 
2^{O(k \log k)} \cdot \Time(\Main(k - 1)). 
\end{align*}
We can conclude $\Time(\Main(k)) \leq 2^{O(k^2 \log k)} \cdot  f(k) \cdot (g(n) + n^4 \log^3 n)$.
\end{proof}


\section{An Algorithm for \Laminarkcut{k}}
\label{sec:laminar}

Recall the definition of the \Laminarkcut{k} problem:
\LamDef*

Let $\solns_{\e_1}$ contain all partitions $S_1,\ldots,S_k$ of $V$ with
the restriction that the boundaries of the first $k-1$ parts is
small---i.e., $w(\partial{S_i}) \le (1+\e_1)\mincut(G)$ for all $i \in
[k-1]$. We emphasize that the weight of the last cut, i.e., 
$w(\partial{S_k})$, is unconstrained. In this section, we give an
algorithm to find a $k$-partition (possibly not in $\solns_{\e_1}$) with
total weight
\[ w(E(S_1,\ldots,S_k)) \le (2 - \eps_2) \min\limits_{\{S_i'\}\in
  \solns_{\e_1}} w(E(S_1',\ldots,S_k')). \]

Formally, the main theorem of this section is the following:
\begin{theorem}[Laminar Cut Algorithm]
  \label{thm:laminar}
  Suppose there exists a $(1+\delta)$-approximation algorithm for
  $\PartialVC{k}$ for some $\delta \in (0, 1/24)$ that runs in time
  $f(k) \cdot g(n)$.  Then, for any $\e_1\in(0,1/6-4\delta)$, there
  exists a $(2 - \eps_2)$-approximation algorithm for
	$\Laminarcut{k}{\eps_1}$ that runs in time $2^{O(k)}f(k)(\tilde O(n^4) + g(n))$
  for some constant $\eps_2 > 0$.
\end{theorem}

In the rest of this section we present the algorithm and the analysis. 
For a formal description, see the pseudocode in
Appendix~\ref{sec:pseudocode-laminar}. 

\subsection{Mincut Tree}

The first idea in the algorithm is to consider the structure of a laminar family of cuts. Below, we introduce the concept of a \textit{mincut tree}. The vertices of the mincut tree are called \textit{nodes}, to distinguish them from the vertices of the original graph.

\begin{definition}[Mincut Tree]
A tree $\mT = (V_{\mT}, E_{\mT}, w_{\mT})$ is a \textbf{$(1+\e_1)$-mincut tree} on a graph $G=(V,E,w)$ with mapping $\phi : V \to V_{\mT}$ if the following two sets are equivalent:
\begin{enumerate}
\item The set of all $(1+\e_1)$-mincuts of $G$.
\item Cut a single edge $e \in E_{\mT}$ of the tree, and let $A_e
  \subset V_{\mT}$ be the nodes on one side of the cut. Define $S_e :=
  \phi^{-1}(A_e) = \{v \mid \phi(v) \in A_e\}$ for each $e\in E_{\mT}$,
  and take the set of cuts $ \{ (S_e, V \setminus S_e) : e \in E_{\mT}
  \}$.
\end{enumerate}
Moreover, for every pair of corresponding $(1+\e_1)$-mincut $(S_e, V
        \setminus S_e)$ and edge $e \in E_{\mT}$, we have $w_{\mT}(e) = w(E(S_e,
        V\setminus S_e))$.
\end{definition}

We use the term \textit{mincut tree} without the $(1+\e_1)$ when the
value of $\e_1$ is either implicit or irrelevant. 

For the rest of this section, let 
\[ \minkut := \mincut(G) \] for brevity. Observe that the last condition
implies that $\minkut\le w_{\mT}(e)\le(1+\e_1)\minkut$ for all $e\in
E_{\mT}$. The existence of a mincut tree (and the algorithm for it)
assuming laminarity, is standard, going back at least to Edmonds and
Giles~\cite{EG75}.

\begin{theorem}[Mincut Tree Existence/Construction]
\label{thm:mincutTreeExistence}
  If the set of $(1+\e_1)$-mincuts of a graph is laminar, then an
  $O(n)$-sized $(1+\e_1)$-mincut tree always exists, and can
  be found in $O(n^3)$ time.
\end{theorem}

\begin{proof}
  We refer the reader to~\cite[Section~2.2]{KV12}. Fix a vertex $v \in
  V$, and for each $(1+\e_1)$-mincut $(S,V\setminus S)$, pick the side
  that contains $v$; this family of subsets of $V$ satisfies the laminar
  condition in Proposition~2.12 of that book. Corollary~2.15 proves that
  this family has size $O(n)$, and the construction of $T$ in
  Proposition~2.14 gives the desired mincut tree. Furthermore, we can
  compute the mincut tree in $O(n^3)$ time as follows: first precompute
  whether $X \subset Y$ for every two sets $X$ and $Y$ in the family,
  and then compute $T$ following the construction in the proof of
  Proposition~2.14.
\end{proof}

\begin{definition}[Mincut Tree Terminology] Let $\mT$ be a rooted mincut
  tree. For $a \in V_{\mT}$, define the following terms:
\begin{OneLiners}
\item[1.] $\children(a)$: the set of children of node $a$ in the rooted tree.
\item[2.] $\desc(a)$: the set of descendants of $a$, i.e., nodes $b \in V_{\mT} \setminus a$ whose path to the root includes $a$.
\item[3.] $\anc(a)$: the set of ancestors of $a$, i.e., nodes $b \in V_{\mT} \setminus a$ on the path from $a$ to the root.
\item[4.] $\subtree(a)$: vertices in the subtree rooted at $a$, i.e., $\{a\} \cup \desc(a)$.
\end{OneLiners}
\end{definition}

For the set of partitions $\solns_{\e_1}$ (as defined at the beginning
of this section), we observe the following.

\begin{claim}[Representing Laminar Cuts in $\mT$]
  Let $\mT = (V_{\mT}, E_{\mT}, w_{\mT})$ be a $(1+\e_1)$-mincut tree of
  $G=(V,E,w)$, and consider a partition $\{S_1, \ldots, S_k\} \in
  \solns_{\e_1}$. Then, there exists a root $r \in V_{\mT}$ and nodes
  $a_1, \ldots, a_{k-1} \in V_{\mT} \setminus r$ such that if we root
  the tree $\mT$ at $r$,
  \begin{enumerate}
  \item For any two nodes in $\{a_1, \ldots, a_{k-1}\}$, neither is an
    ancestor of the other. (We call two such nodes
    \textbf{incomparable}).
  \item For each $v_i$, let $A_i := \subtree(a_i)$, and let $A_k =
    V_{\mT} \setminus \bigcup_{i=1}^{k-1} A_i$ (so that $r \in A_k$). We
    have the two equivalences $\{\phi^{-1}(A_i) \mid i \in [k-1]\} = \{S_1,
    \ldots, S_{k-1}\}$ and $\phi^{-1}(A_k) = S_k$. In other words, the
    components $A_i \subset V_{\mT}$, when mapped back by $\phi^{-1}$,
    correspond exactly to the sets $S_i \subset V$, with the additional
    guarantee that $A_k$ and $S_k$ match. 
\end{enumerate}
\end{claim}

\begin{proof}
  Since $S_i$ is a $(1+\e_1)$-mincut for each $i \in [k-1]$, there
  exists an edge $e_i \in E_{\mT}$ such that the set $A_i'$ of nodes on
  one side of $e_i$ satisfies $\phi^{-1}(A_i') = S_i $. The sets $A_i'$
  for $i\in[k-1]$ are necessarily disjoint, and they cannot span all
  nodes in $V_{\mT}$, since $S_k$ is still unaccounted for. If we root
  $\mT$ at a node $r$ not in any $A_i'$, then each $A_i'$ is a subtree
  of the rooted $\mT$. Altogether, the roots of the subtrees $A_i'$
  satisfy condition~(1) of the lemma, and the $A_i'$ themselves satisfy
  condition~(2).
\end{proof}

For a graph $G=(V,E,w)$ and mincut tree $\mT=(V_{\mT},E_{\mT},w_{\mT})$
with mapping $\phi:V\to V_{\mT}$, define $E_G(A,B)$ for $A,B \subset
V_{\mT}$ as $E\left(\phi^{-1}(A), \phi^{-1}(B)\right)$, i.e., the total
weight of edges crossing the sets corresponding to $A$ and $B$ in $V$.

\begin{observation}
  Given a root $r\in V_{\mT}$ and incomparable nodes $a_1, \ldots,
  a_{k-1} \in V_{\mT} \setminus r$, we can bound the corresponding
  partition $S_1,\ldots,S_k$ as follows:
  \begin{align*}
    w(E(S_1, \ldots, S_k)) &= \textstyle \sum_{i=1}^{k-1}
    w(\partial(S_i)) - \sum_{i<j \le k-1} w(E(S_i, S_j)) \\ &=
    \textstyle \sum_{i=1}^{k-1} w_{\mT}(e_i) - \sum_{i<j\le k-1}
    w(E_G(\subtree(a_i),\subtree(a_j))) ,
  \end{align*}
  where $e_i$ is the parent edge
  of $v_i$ in the rooted tree.
\end{observation}

Note that $\minkut \le w_{\mT}(e) \le (1+\e_1)\minkut$ for all $e\in
E_{\mT}$, so to approximately minimize the above expression for a fixed
root $r$, it suffices to approximately maximize 
\begin{align*}
 \textstyle \saved(a_1,\ldots,a_{k-1}) := \sum\limits_{i<j\le k-1}
  w(E_G(\subtree(a_i),\subtree(a_j))), \label{eq:saved}
\end{align*}
which we think of as the edges \textit{saved} in the double counting of
$\sum_{i=1}^{k-1}w_{\mT}(e_i)$.  The actual approximation factor is made
precise in the proof of Theorem~\ref{thm:laminar}.

To maximize the number of saved edges over all partitions in
$\solns_{\e_1}$, it suffices to try all possible roots $r$ and take the
best partition.  Therefore, for the rest of this section, we focus on
maximizing $\saved(a_1,\ldots,a_{k-1})$ for a fixed root
$r$. 
Let $\ell^*(r)$ be that maximum value for root $r$, and let $\opt(r) =
\{a_1^*, \ldots, a_{k-1}^*\} \subset V_{\mT}$ be the solution that
attains it.

\subsection{Anchors}

Root the mincut tree $\mT$ at $r$, and let $a_1^*,\ldots,a_{k-1}^*$ be incomparable nodes in the solution $\Opt(r)$.
First, observe that we can assume w.l.o.g.\ that for each node $a_i^*$,
its parent node is an ancestor of some $a_j^* \neq a_i^*$: if not, we can replace $a_i^*$ with its parent, which can only increase $\saved(a_1^*,\ldots,a_{k-1}^*)$.

\begin{observation}
Consider nodes $a_1^*,\ldots,a_s^* \in \opt(r)$ which share the same parent $a \notin \opt(r)$, and assume that $a$ has no other descendants. If we replace $a_1^*,\ldots,a_s^*$ in $\opt(r)$ with $a$, then we lose at most $\saved(a_1^*,\ldots,a_s^*)$ in our solution.\footnote{The new solution may no longer have $k-1$ nodes, but we will fix this problem in the proof of Theorem~\ref{thm:laminar}. For now, assume that we are allowed to choose any number up to $k-1$ nodes.}
\end{observation}

If $\saved(a_1^*,\ldots,a_s^*)$ is small, i.e., compared to $(s-1)\minkut$, then we do not lose too much. This idea motivates the idea of anchors.

\begin{definition}[Anchors]
Let $\mT=(V_{\mT},E_{\mT},w_{\mT})$ be a rooted tree.
For a fixed constant $\e_3 > 0$, define an \textbf{$\e_3$-anchor} to be a node $a\in V_{\mT}$ such that there exists $s \in [2,k-1]$ and $s$ children $a_1,\ldots,a_s$ such that $\saved(a_1,\ldots,a_s) \ge \e_3(s-1)\minkut$. When the value of $\e_3$ is implicit, we use the term anchor, without the $\e_3$.
\end{definition}

We now claim that we can transform any solution to another
well-structured solution, with only a minimal loss.

\begin{lemma}[Shifting Lemma]
  \label{lem:anchor}
  Let $a_1,\ldots,a_{k-1}$ be a set of incomparable nodes of a
  $(1+\e_1)$-mincut tree $\mT$.  Then, there exists a set
  $b_1,\ldots,b_s$ of incomparable nodes, for $1\le s\le k-1$, such that
\begin{enumerate}
\item The parent of every node $b_i$ is either an $\e_3$-anchor, or is
  an ancestor of some node $b_j \neq b_i$ whose parent is an anchor.
\item $\saved(b_1,\ldots,b_s) \ge \saved(a_1,\ldots,a_{k-1}) -
  \e_3(k-s)\minkut$. 
\end{enumerate}
In particular, if $\{a_1, \ldots, a_{k-1}\} = \opt(r)$, condition~(2)
implies $\saved(b_1,\ldots,b_s) \ge \ell^*(r) - \e_3(k-1)\minkut$.
\end{lemma}

\begin{proof}
  We begin with the solution $b_i=a_i$ for all $i$, and iteratively
  shift non-anchors in the solution while maintaining the potential
  function $\Phi:=\saved(b_1,\ldots,b_s) - \saved(a_1,\ldots,a_{k-1}) +
  \e_3(k-s)\minkut$ nonnegative.  
  At the beginning, $\Phi=0$. Suppose there is a
  node $b_i$ not satisfying condition (1). Choose one such $b_i$ of
  maximum depth in the tree, and let $b'$ be its non-anchor parent. Then
  the only descendants of $b'$ in the current solution are siblings of
  $b_i$. Replace $b_i$ and its $s'$ siblings in the solution by
  $b'$. Since $b'$ is not an anchor, $\saved(b_1,\ldots,b_s)$ drops by
  at most $\e_3(s'-1)\minkut$. This drop is compensated by the decrease
  of the solution size from $s$ to $s-(s'-1)$.
\end{proof}

Hence, at a loss of $\e_3(k-1)\minkut$, it suffices to focus on a
solution $\opt'(r)$ which fulfills condition (1) of
Lemma~\ref{lem:anchor} and has $\saved$ value $\ell'(r) \geq \ell^*(r) -
\e_3(k-1)\minkut$.

The rest of the algorithm splits into two cases. At a high level, if
there are enough anchors in a mincut tree $\mT$ that are incomparable
with each other, then we can take such a set and be done. Otherwise, the
set of anchors can be grouped into a small number of paths in $\mT$, and
we can afford to try all possible arrangements of anchors. But first we
show how to find all the anchors in $\mT$.

\subsection{Finding Near-Anchors}

\newcommand{\anchors}{\mathcal{A}}

\begin{lemma}[Finding (Near-)Anchors]
  \label{lem:compute-anchors}
  Assume access to a $(1+\delta)$-approximation algorithm for $\PartialVC k$
  running in time $f(k) \cdot g(n)$. Then, there is an algorithm running
  in time $O(n \cdot (n^2 + k \cdot f(k) \cdot g(n)))$ that computes a
  set $\anchors$ of ``near''-anchors in $\mT$, i.e., vertices $a \in
  V_{\mT}$ for which there exists an integer $s \in [2,k-1]$ and $s$
  children $b_1, \ldots, b_s$ such that $\saved(b_1,\ldots,b_s) \geq
  \e_3(s-1)\minkut - \delta(1+\e_1)s\minkut$.
\end{lemma}

\begin{proof}
  To determine if a node $a$ is an anchor or not, for each integer $s
  \in [2, k-1]$ we wish to compute the maximum value of
  $\saved(b_1,\ldots,b_s)$ for $b_1,\ldots,b_s \in
  \children(a)$. Consider the following weighted, complete graph with
  vertex and edge weights: for each $b \in \children(a)$ create a vertex
  $x_b$, and the edge $(x_{b_1}, x_{b_2})$ has weight
  $\saved(b_1,b_2)$. Each vertex $x_b$ also has weight $(1+\e_1)\minkut
  - w(\partial x_b)$, where $w(\partial x_b)$ is the sum of the weights
  of edges incident to $x_b$. Note that this graph is
  $(1+\e_1)\minkut$-regular, if we include vertex weights in the
  definition of vertex degree.

  Observe that $w(\partial x_b)
  \le \partial\left(\phi^{-1}(\subtree(b))\right) \le (1+\e_1)\minkut$,
  since every edge in $G$ that contributes to $\saved(b,b')$ for another
  child $b'$ also contributes to the cut
  $\partial\left(\phi^{-1}(\subtree(b))\right)$, which we know is $\le
  (1+\e_1)\minkut$.
  Therefore, each vertex has a nonnegative weight.
  Also, a partial vertex cover on this graph with vertices $x_{b_1},
  \ldots, x_{b_s}$ has weight exactly $(1+\e_1)s\minkut - \saved (b_1,
  \ldots, b_s)$.

  Let $b_1^*,\ldots,b_s^* \in \children(a)$ be the solution with maximum
  $\saved(b_1^*,\ldots,b_s^*)$.  To compute this maximum, we can build
  the above graph and run the $(1+\delta)$-approximate partial vertex
  cover algorithm from Theorem~\ref{thm:pvc}. The solution
  $b_1,\ldots,b_s$ satisfies
  \[ (1+\e_1)s\minkut - \saved(b_1,\ldots,b_s) \le
  (1+\delta)\left((1+\e_1)s\minkut - \saved(b_1^*,\ldots,b_s^*)\right),
  \]
  so that
  \begin{align*}
    \saved(b_1,\ldots,b_s) & \ge (1+\delta)\,\saved(b_1^*,\ldots,b_s^*)
    - \delta(1+\e_1)s \minkut \\
    & \ge \saved(b_1^*,\ldots,b_s^*) - \delta(1+\e_1)s\minkut.
  \end{align*}

  We run this subprocedure for the vertex $a$ for each integer $2 \leq s
  \le \min\{|\children(a)|, k-1\}$, and mark vertex $a$ if there exists
  an integer $s$ such that the weight of saved edges is at least
  $\e_3(s-1)\minkut - \delta(1+\e_1)s\minkut$. The set $\anchors$ of
  near-anchors is exactly the set of marked vertices.

        As for running time, for each node $a$, it takes $O(n^2)$ time to construct the $\pvc$ graph and $O(k) \cdot f(k) \cdot g(n)$ time to solve $\PartialVC s$ for each $s \in [2,k-1]$. Repeating the above for each of the $O(n)$ nodes achieves the promised running time.
\end{proof}

\subsection{Many Incomparable Near-Anchors}

\begin{lemma}[Many Anchors]
  \label{lem:incomparable}
  Suppose we have access to a $(1+\delta)$-approximation algorithm for
  $\PartialVC k$ running in time $f(k) \cdot g(n)$. Suppose the set
  $\anchors$ of near-anchors contains $k-1$ incomparable nodes from the
  mincut tree $\mT$. Then, there is an algorithm computing a solution
  with $\saved$ value
  $\ge \frac14\e_3(k-1)\minkut - \delta(1+\e_1)(k-1)\minkut$ for any
  $\delta>0$, running in time
  $O(n \cdot (n^2 + k \cdot f(k) \cdot g(n)))$.
\end{lemma}

\begin{proof}
  First, we compute the set $\anchors$ in $O(n \cdot (n^2 + k \cdot f(k)
  \cdot g(n))$ time, according to Lemma~\ref{lem:compute-anchors}. If
  $\anchors$ contains $k-1$ incomparable nodes, we can \textit{find}
  them in $O(n^2)$ time by greedily choosing nodes in a topological,
  bottom-first order (see lines 4--11 in
  Algorithm~\ref{alg:laminarRooted}). Each of these $k-1$ marked nodes
  $a_1, \ldots, a_{k-1}$ has an associated value $s_i$, indicating that
  $a_i$ has some $s_i$ children whose $\saved$ value is at least
  $\e_3(s_i-1)\minkut - \delta(1+\e_1)s_i\minkut$. If we consider a
  subset $A \subset [k-1]$ and choose the $s_i$ children for each $a_i$
  with $i\in A$, then we get a set with $\sum_{i\in A} s_i$ nodes, whose
  total $\saved$ value at least \[\e_3\left( \sum_{i\in
      A}(s_i-1)\right)\minkut - \delta(1+\e_1)\left( \sum_{i\in A}s_i
  \right)\minkut.\] Assuming that $\sum_{i\in A}s_i \le k-1$, i.e., we
  choose at most $k-1$ children, the second
  $\delta(1+\e_1)\left(\sum_{i\in A}s_i\right)\minkut$ term is at most
  $\delta(1+\e_1) (k-1)\minkut$. To optimize the $\e_3\left( \sum_{i\in
      A}(s_i-1)\right)\minkut$ term, we reduce to the following knapsack
  problem: we have $k-1$ items $i \in [k-1]$ where item $i$ has size
  $s_i \in [2,k-1]$ and value $s_i-1$, and our bag size is $k-1$. A
  knapsack solution of value $Z:=\sum_{i\in A}(s_i-1)$ translates to a
  solution with $\saved$ value $\ge \e_3\minkut \cdot Z -
  \delta(1+\e_1)(k-1)\minkut$. By Lemma~\ref{lemma:knapsack}, when $k
  \ge 5$, we can compute a solution $A \subset [k-1]$ of value $\ge
  (k-1)/4$ in $O(k)$ time. (If $k\le4$, we can use the exact $\tilde O(n^4)$
  \kcut algorithm from~\cite{Levine00}.) Selecting the children of each
  $u_i$ with $i\in A$ gives a total $\saved$ value of at least
  $\frac14\e_3(k-1)\minkut - \delta(1+\e_1)(k-1)\minkut$.
\end{proof}

\subsection{Few Incomparable Near-Anchors}

\begin{figure}
\centering
\begin{tikzpicture} [xscale=.5, yscale=.3]
\tikzstyle{every node}=[circle, fill, scale=.3];
\node (v1) at (-1,4) {};
\node (v2) at (-2,2) {};
\node (v12) at (0.5,2) {};
\node (v3) at (-3,0) {};
\node (v10) at (-1.5,0) {};
\node (v11) at (0,0) {};
\node (v13) at (1,0) {};
\node (v14) at (2,0) {};
\node (v4) at (-4,-2) {};
\node (v7) at (-2,-2) {};
\node (v5) at (-5,-4) {};
\node (v6) at (-3.5,-4) {};
\node (v8) at (-2.5,-4) {};
\node (v9) at (-1,-4) {};
\node (v15) at (1,-2) {};
\node (v16) at (3,-2) {};
\draw [line width=.5pt] (v1) edge (v2);
\draw [line width=.5pt] (v2) edge (v3);
\draw [line width=.5pt] (v3) edge (v4);
\draw [line width=.5pt] (v4) edge (v5);
\draw [line width=.5pt] (v4) edge (v6);
\draw [line width=.5pt] (v7) edge (v3);
\draw [line width=.5pt] (v8) edge (v7);
\draw [line width=.5pt] (v9) edge (v7);
\draw [line width=.5pt] (v10) edge (v2);
\draw [line width=.5pt] (v11) edge (v12);
\draw [line width=.5pt] (v12) edge (v1);
\draw [line width=.5pt] (v13) edge (v12);
\draw [line width=.5pt] (v12) edge (v14);
\draw [line width=.5pt] (v15) edge (v14);
\draw [line width=.5pt] (v14) edge (v16);
\draw (v3) circle (.3);
\draw (v14) circle (.3);
\draw (v16) circle (.3);
\draw (v5) circle (.3);
\draw (v8) circle (.3);
\draw (v9) circle (.3);
\draw [line width=.5pt] (v5) edge (-5.5,-5);
\draw [line width=.5pt] (v5) edge (-4.5,-5);
\draw [line width=.5pt] (v8) edge (-3,-5);
\draw [line width=.5pt] (v8) edge (-2.25,-5);
\draw [line width=.5pt] (v9) edge (-1.25,-5);
\draw [line width=.5pt] (v9) edge (-0.5,-5);
\end{tikzpicture}
\qquad
\begin{tikzpicture} [xscale=.5, yscale=.3]
\tikzstyle{every node}=[circle, fill, scale=.5];
\node (v1) at (-1,4) {};
\node (v3) at (-3,0) {};
\node (v7) at (-2,-2) {};
\node (v5) at (-5,-4) {};
\node (v8) at (-2.5,-4) {};
\node (v9) at (-1,-4) {};
\node (v16) at (3,-2) {};
\draw [white,line width=.5pt] (v5) edge (-5.5,-5);
\draw [white,line width=.5pt] (v5) edge (-4.5,-5);
\draw [white,line width=.5pt] (v8) edge (-3,-5);
\draw [white,line width=.5pt] (v8) edge (-2.25,-5);
\draw [white,line width=.5pt] (v9) edge (-1.25,-5);
\draw [white,line width=.5pt] (v9) edge (-0.5,-5);
\draw [line width=1pt] (v1) edge (v3);
\draw [line width=1pt] (v3) edge (v5);
\draw [line width=1pt] (v3) edge (v7);
\draw [line width=1pt] (v7) edge (v8);
\draw [line width=1pt] (v9) edge (v7);
\draw [line width=1pt] (v1) edge (v16);
\end{tikzpicture}
\qquad
\begin{tikzpicture} [xscale=.5, yscale=.3]
\tikzstyle{every node}=[circle, fill, scale=.3];
\node[fill=brown] (v1) at (-1,4) {};
\node[fill=blue] (v2) at (-2,2) {};
\node[fill=green] (v12) at (0.5,2) {};
\node[fill=blue] (v3) at (-3,0) {};
\node[fill=] (v10) at (-1.5,0) {};
\node[fill=] (v11) at (0,0) {};
\node[fill=] (v13) at (1,0) {};
\node[fill=green] (v14) at (2,0) {};
\node[fill=red] (v4) at (-4,-2) {};
\node[fill=purple] (v7) at (-2,-2) {};
\node[fill=red] (v5) at (-5,-4) {};
\node[fill=] (v6) at (-3.5,-4) {};
\node[fill=orange] (v8) at (-2.5,-4) {};
\node[fill=yellow] (v9) at (-1,-4) {};
\node[fill=] (v15) at (1,-2) {};
\node[fill=green] (v16) at (3,-2) {};
\draw [blue,line width=1pt] (v1) edge (v2);
\draw [blue,line width=1pt] (v2) edge (v3);
\draw [red,line width=1pt] (v3) edge (v4);
\draw [red,line width=1pt] (v4) edge (v5);
\draw [line width=1pt] (v4) edge (v6);
\draw [purple,line width=1pt] (v7) edge (v3);
\draw [orange,line width=1pt] (v8) edge (v7);
\draw [yellow,line width=1pt] (v9) edge (v7);
\draw [line width=1pt] (v10) edge (v2);
\draw [line width=1pt] (v11) edge (v12);
\draw [green,line width=1pt] (v12) edge (v1);
\draw [line width=1pt] (v13) edge (v12);
\draw [green,line width=1pt] (v12) edge (v14);
\draw [line width=1pt] (v15) edge (v14);
\draw [green,line width=1pt] (v14) edge (v16);
\draw [white,line width=.5pt] (v5) edge (-5.5,-5);
\draw [white,line width=.5pt] (v5) edge (-4.5,-5);
\draw [white,line width=.5pt] (v8) edge (-3,-5);
\draw [white,line width=.5pt] (v8) edge (-2.25,-5);
\draw [white,line width=.5pt] (v9) edge (-1.25,-5);
\draw [white,line width=.5pt] (v9) edge (-0.5,-5);
\end{tikzpicture}
\caption{\label{figure:branches}
Establishing the set of branches $\mB$. The circled nodes on the left are the near-anchors. The middle graph is the tree $\mT'$. On the right, each non-black color is an individual branch; actually, the branches only consist of nodes, but we connect the nodes for visibility. Also, note that the root is its own branch. The red, orange, yellow, and green branches form an incomparable set.}
\end{figure}
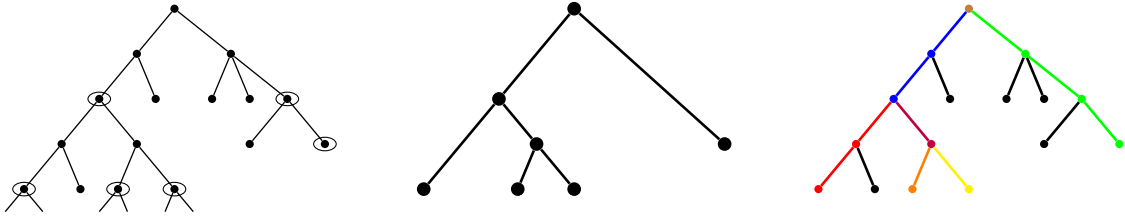

If the condition in Lemma~\ref{lem:incomparable} does not hold, then
there exist $\le k-2$ paths from the root in $\mT$ such that every node
in the near-anchor set $\anchors$ lies on one of these paths. If we view
the union of these paths as a tree $\mT'$ with $\le k-2$ leaves, then we
can partition the nodes in tree $\mT'$ into a collection $\mB$ of at
most $2k-3$ \textit{branches}. Each branch $B$ is a collection of
vertices obtained by taking either a leaf of $\mT'$ or a vertex of
degree more than two, and all its immediate degree-2 ancestors; see
Figure~\ref{figure:branches}. Note that it is possible that the root
node is its own branch. Hence, given two branches $B_1, B_2 \in \mB$,
either every node from $B_1$ is an ancestor of every node from $B_2$ (or
vice versa), or else every node from $B_1$ is incomparable with every
node from $B_2$.

Let $A' \sse \anchors$ be the set of anchors with at least one child in
$\opt'(r)=\{a_1^*,\ldots,a_s^*\}$; recall that $\opt'(r)$ was produced
by the shifting procedure in Lemma~\ref{lem:anchor}. Let $A^* \sse A'$
be
the \textit{minimal} anchors in $A'$, i.e., every anchor in $A'$ that is
not an ancestor of any other anchor in $A'$. We know that every anchor in
$A^*$ falls inside our set of branches, although the algorithm does not
know where.  Moreover, by condition~(1) of Lemma~\ref{lem:anchor}, the
parent of every $a_i^* \in \opt'(r)$ either lies in $A^*$, or is an
ancestor of an anchor in $A^*$.

As a warm-up, consider the case where all the anchors in $A'$ are
contained within a single branch.

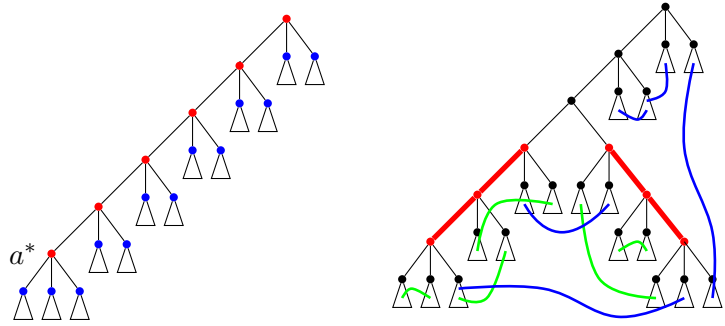
\begin{figure} 
\centering

\begin{tikzpicture}[scale=0.25]
\node [circle, fill=red, scale=.3] (v3) at (0,5) {};
\node [circle, fill=red, scale=.3] (v2) at (2.5,7.5) {};
\node [circle, fill=red, scale=.3] (v1) at (5,10) {};
\node [circle, fill=red, scale=.3] (v4) at (-2.5,2.5) {};
\node [circle, fill=red, scale=.3] (v5) at (-5,0) {};
\node [circle, fill=red, scale=.3, label=left:$a^*$] (v6) at (-7.5,-2.5) {};
\draw  (v1) edge (v2);
\draw  (v2) edge (v3);
\draw  (v3) edge (v4);
\draw  (v4) edge (v5);
\draw  (v5) edge (v6);
\node [circle, fill=blue, scale=.3] (v7) at (5,8) {};
\node [circle, fill=blue, scale=.3] (v8) at (6.5,8) {};
\node [circle, fill=blue, scale=.3] (v9) at (2.5,5.5) {};
\node [circle, fill=blue, scale=.3] (v10) at (4,5.5) {};
\node [circle, fill=blue, scale=.3] (v11) at (0,3) {};
\node [circle, fill=blue, scale=.3] (v12) at (1.5,3) {};
\node [circle, fill=blue, scale=.3] (v13) at (-2.5,0.5) {};
\node [circle, fill=blue, scale=.3] (v14) at (-1,0.5) {};
\node [circle, fill=blue, scale=.3] (v15) at (-5,-2) {};
\node [circle, fill=blue, scale=.3] (v16) at (-3.5,-2) {};
\node [circle, fill=blue, scale=.3] (v17) at (-9,-4.5) {};
\node [circle, fill=blue, scale=.3] (v18) at (-7.5,-4.5) {};
\node [circle, fill=blue, scale=.3] (v19) at (-6,-4.5) {};
\draw (v1) -- (v7) -- (4.5,6.5) -- (5.5,6.5) -- (v7) (v1) -- (v8) -- (6,6.5) -- (7,6.5) -- (v8) (v2) -- (v9) -- (2,4) -- (3,4) -- (v9) (v2) -- (v2) -- (v10) -- (3.5,4) -- (4.5,4) -- (v10) (v3) -- (v11) -- (-0.5,1.5) -- (0.5,1.5) -- (v11) (v3) -- (v12) -- (1,1.5) -- (2,1.5) -- (v12) (v4) -- (v13) -- (-3,-1) -- (-2,-1) -- (v13) (v4) -- (v14) -- (-1.5,-1) -- (-0.5,-1) -- (v14) (v5) -- (v15) -- (-5.5,-3.5) -- (-4.5,-3.5) -- (v15) (v5) -- (v16) -- (-4,-3.5) -- (-3,-3.5) -- (v16) (v6) -- (v17) -- (-9.5,-6) -- (-8.5,-6) -- (v17) (v6) -- (v18) -- (-8,-6) -- (-7,-6) -- (v18) (v6) -- (v19) -- (-6.5,-6) -- (-5.5,-6) -- (v19);
\end{tikzpicture}
\qquad
\begin{tikzpicture}[scale=0.25]
\node [circle, fill=black, scale=.3] (v3) at (0,5) {};
\node [circle, fill=black, scale=.3] (v2) at (2.5,7.5) {};
\node [circle, fill=black, scale=.3] (v1) at (5,10) {};
\node [circle, fill=red, scale=.3] (v4) at (-2.5,2.5) {};
\node [circle, fill=red, scale=.3] (v5) at (-5,0) {};
\node [circle, fill=red, scale=.3] (v6) at (-7.5,-2.5) {};
\draw  (v1) edge (v2);
\draw  (v2) edge (v3);
\draw  (v3) edge (v4);
\draw[red, line width=2pt]  (v4) edge (v5);
\draw[red, line width=2pt]  (v5) edge (v6);
\node [circle, fill=black, scale=.3] (v7) at (5,8) {};
\node [circle, fill=black, scale=.3]  (v8) at (6.5,8) {};
\node [circle, fill=black, scale=.3] (v9) at (2.5,5.5) {};
\node [circle, fill=black, scale=.3] (v10) at (4,5.5) {};

\node [circle, fill=black, scale=.3] (v13) at (-2.5,0.5) {};
\node [circle, fill=black, scale=.3](v14) at (-1,0.5) {};
\node [circle, fill=black, scale=.3](v15) at (-5,-2) {};
\node [circle, fill=black, scale=.3] (v16) at (-3.5,-2) {};
\node [circle, fill=black, scale=.3] (v17) at (-9,-4.5) {};
\node [circle, fill=black, scale=.3](v18) at (-7.5,-4.5) {};
\node [circle, fill=black, scale=.3](v19) at (-6,-4.5) {};
\draw (v1) -- (v7) -- (4.5,6.5) -- (5.5,6.5) -- (v7);
\draw (v1) -- (v8) -- (6,6.5) -- (7,6.5) -- (v8);
\draw (v2) -- (v9) -- (2,4) -- (3,4) -- (v9);
\draw (v2) -- (v10) -- (3.5,4) -- (4.5,4) -- (v10);
\draw (v4) -- (v13) -- (-3,-1) -- (-2,-1) -- (v13) (v4) -- (v14) -- (-1.5,-1) -- (-0.5,-1) -- (v14) (v5) -- (v15) -- (-5.5,-3.5) -- (-4.5,-3.5) -- (v15) (v5) -- (v16) -- (-4,-3.5) -- (-3,-3.5) -- (v16) (v6) -- (v17) -- (-9.5,-6) -- (-8.5,-6) -- (v17) (v6) -- (v18) -- (-8,-6) -- (-7,-6) -- (v18) (v6) -- (v19) -- (-6.5,-6) -- (-5.5,-6) -- (v19);
\node [circle, fill=red, scale=.3] (v11) at (2,2.5) {};
\node [circle, fill=red, scale=.3] (v21) at (4,0) {};
\node [circle, fill=red, scale=.3] (v28) at (6,-2.5) {};
\node [circle, fill=black, scale=.3](v12) at (0.5,0.5) {};
\node [circle, fill=black, scale=.3](v20) at (2,0.5) {};
\node [circle, fill=black, scale=.3](v22) at (2.5,-2) {};
\node[circle, fill=black, scale=.3] (v23) at (4,-2) {};
\node[circle, fill=black, scale=.3] (v26) at (6,-4.5) {};
\node [circle, fill=black, scale=.3](v24) at (4.5,-4.5) {};
\node [circle, fill=black, scale=.3](v27) at (7.5,-4.5) {};
\draw [red, line width=2pt] (v11) edge (v21);
\draw [red, line width=2pt] (v21) edge (v28);
\draw  (v3) edge (v11);
\draw (v11) -- (v12) -- (0,-1) -- (1,-1) -- (0.5,0.5);
\draw (v11) -- (v20) -- (1.5,-1) -- (2.5,-1) -- (v20);
\draw (v21) -- (v22) -- (2,-3.5) -- (3,-3.5) -- (v22) (v21) -- (v23) -- (3.5,-3.5) -- (4.5,-3.5) -- (v23);
\draw (v28) -- (v24) -- (4,-6) -- (5,-6) -- (v24) (v28) -- (v26) -- (5.5,-6) -- (6.5,-6) -- (6,-4.5) (6,-2.5) -- (v27) -- (7,-6) -- (8,-6) -- (v27);

\draw [green, line width=1pt] plot[smooth, tension=.7] coordinates {(-9,-5.5) (-8.5,-5) (-7.5,-5.5)};
\draw [green, line width=1pt] plot[smooth, tension=.7] coordinates {(-6,-5.5) (-4.5,-5.5) (-3.5,-3)};
\draw [green, line width=1pt] plot[smooth, tension=.7] coordinates {(-5,-3) (-4,-0.5)  (-1,-0.5)};
\draw [green, line width=1pt] plot[smooth, tension=.7] coordinates {(0.5,-0.5) (1.5,-4.5) (4.5,-5.5)};
\draw [green, line width=1pt] plot[smooth, tension=.7] coordinates {(2.5,-3) (3.5,-2.5) (4,-3)};
\draw [blue, line width=1pt] plot[smooth, tension=1] coordinates {(-2.5,-0.5) (-0.5,-2) (2,-0.5)};
\draw [blue, line width=1pt] plot[smooth, tension=.7] coordinates {(-6,-5) (-1.5,-5) (2.5,-6.5) (6,-5.5)};
\draw [blue, line width=1pt] plot[smooth, tension=.7] coordinates {(2.5,4.5) (3.5,4) (4,4.5)};
\draw [blue, line width=1pt] plot[smooth, tension=.7] coordinates {(6.5,7) (6,2.5) (7.5,-2) (7.5,-5.5)};
\draw [blue, line width=1pt] plot[smooth, tension=.7] coordinates {(4,5) (5,5.5) (5,7)};
\end{tikzpicture}
\caption{\label{figure:single-branch}Left (Claim~\ref{claim:singleBranch}): The red nodes form our branch $B$, and the blue nodes form the set $\children((\{a^*\}
  \cup \anc(a^*)) \cap B)$. The triangles are the subtrees participating in the $\pvc$ instance.
Right (Lemma~\ref{lem:chains}): The red nodes form our two incomparable branches. The green edges are internal edges, while the blue edges are external.}
\end{figure}

\begin{claim}[Warm-up]
  \label{claim:singleBranch}
  Assume there exists a $(1+\delta)$-approximation algorithm for $\PartialVC k$ running
in time $f(k) \cdot g(n)$. Suppose the set of anchors $A'$ with at least one child in
  $\opt'(r)$ is contained within a single branch $B$.  Then there is an
  algorithm computing a solution with $\saved$ value at least $\ell'(r)
  - \delta(1+\e_1)(k-1)\minkut$, running in time $O(n \cdot (n^2 + f(k) \cdot  g(n)))$.
\end{claim}

\begin{proof}
  If all of $A'$ lies on $B$, the minimal anchor $a^* \in A^*$ must also
  be in $B$. Moreover, for every $a_i^*\in\opt'(r)$, its parent is
  either $a^*$ or an ancestor of $a^*$, which means that $\opt'(r) \sse
  \children((\{a^*\} \cup \anc(a^*)) \cap B)$. Since the nodes in $\children((\{a^*\}
  \cup \anc(a^*)) \cap B)$ are incomparable (see Figure~\ref{figure:single-branch}), we can construct the same
  graph as the one in Lemma~\ref{lem:compute-anchors} on all these nodes
  in $\children((\{a^*\} \cup \anc(a^*)) \cap B)$ and run the \pvc-based
  algorithm to get the same $\saved$ guarantees (see
  Algorithm~\ref{alg:subtreePVC}). 

  Therefore, the algorithm guesses the location of $a^*$ inside $B$ by
  trying all possible $|B|=O(n)$ nodes, and for each choice of $a^*$,
  runs the $(1-\delta)$-approximate \pvc-based algorithm from
  Lemma~\ref{lem:compute-anchors} on the corresponding graph (see
  Algorithm~\ref{alg:singleBranch}). 
\end{proof}

Now for the general case. Consider $\opt'(r)$ and the set of all
branches $\mB$. Let $\mB^* \sse \mB$ be the incomparable branches that
contain the minimal anchors, i.e., those in $A^*$. We classify the
$\ell(r')$ saved edges in $\opt'(r)$ into two groups (see Figure~\ref{figure:single-branch}): if an edge is
saved between the subtrees below $a_i^*,a_j^* \in \opt'(r)$ whose
parent(s) belong to the same branch in $\mB^*$, then call this
\textit{an internal edge}. Otherwise, it is an \textit{external edge}:
these are saved edges in $\opt'(r)$ that either go between two subtrees
in different branches, or between subtrees in the same branch in $\mB
\setminus \mB^*$. One of the two sets has $\ge \frac12\ell'(r)$ saved
edges, and we provide two separate algorithms, one to approximate each
group.

\begin{lemma}\label{lem:chains}
  Assume there exists a $(1+\delta)$-approximation algorithm for $\PartialVC k$ running
in time $f(k) \cdot g(n)$. Suppose that all anchors of $\opt'(r)$ are contained in a set $\mB$ of
  $\le 2k-3$ branches. Then there is an algorithm that computes a
  solution with $\saved$ value $\ge \frac12\ell'(r) -
  \delta(1+\e_1)(k-1)\minkut$, running in time $ 2^{O(k)} \cdot (n^2+f(k)\cdot g(n)) $.
\end{lemma}

\begin{proof}
  \emph{Case I: internal edges $\ge\frac12\ell'$.} For each branch
  $B\in\mB$ and each $s\in [k-1]$, compute a solution of $s$ nodes that
  maximizes the number of internal edges \textit{within branch
    $\mB$}, in the same manner as in
  Claim~\ref{claim:singleBranch}; this takes time $O(k^2n \cdot (n^2 + f(k) \cdot  g(n)))$. Finally, guess all possible $\le
  2^{2k-3}$ subsets of incomparable branches; for each subset
  $\mB'\subseteq\mB$, try all vectors $\mathbf i \in [k-1]^{\mB'}$ with
  $\sum_{B\in\mB'} i_B \le k-1$, look up the solution using $i_B$
  vertices in branch $B$, and sum up the total number of internal
  edges. Actually, trying all vectors $\mathbf i \in [k-1]^{\mB'}$ takes $k^{O(k)}$ time, but we can speed up this step to $\poly(k)$ time using dynamic programming. Since one
  of the guesses $\mB'$ will be $\mB^*$, the best solution will save at
  $\ge\frac12\ell'(r)-\delta(1+\e_1)(k-1)\minkut$ edges. The total running time for this case is $O(k^2 \cdot f(k) \cdot g(n) + 2^{2k}\cdot\poly(k))$.

  \emph{Case II: external edges $\ge\frac12\ell'$.} Again, we guess the
  set $\mB^*\subset\mB$ of incomparable branches containing minimal
  anchors $A^*$. For a branch $B \in \mB^*$, let $a_B:=(a \in B : B
  \setminus a \sse \desc(a))$ be the ``highest'' node in $B$, that is an
  ancestor of every other node in $B$. For each branch, we can replace
  all nodes in $\opt'(r)$ that are descendants of $a_B$ with just $a_B$;
  doing can only increase the number of external edges. The new solution
  has all nodes contained in the set
  \[ \children\bigg(\anc\bigg(\bigcup_{B\in\mB^*}\{a_B\}\bigg)\bigg), \]
  which is a set of incomparable nodes. Therefore, we can construct the
  graph of Lemma~\ref{lem:compute-anchors} and use the \pvc-based
  algorithm with this node set instead. This gives a solution with $\ge
  \frac12\ell'(r)-\delta(1+\e_1)(k-1)\minkut$ saved edges. The total running time for this case is $O(2^{2k}\cdot (n^2 + f(k) \cdot g(n)))$.
\end{proof}

\subsection{Combining Things Together}

Putting things together, we conclude with Theorem~\ref{thm:laminar}.  We
refer the reader to Algorithm~\ref{alg:laminar} for the pseudocode of
the entire algorithm.

\begin{proof}[Proof (Theorem~\ref{thm:laminar}).]
  Let the original graph be $G=(V,E,w).$ We compute a $(1+\e_1)$-mincut
  tree $\mT=(V_{\mT},E_{\mT},w_{\mT})$ with mapping $\phi:V\to V_{\mT}$ in time $O(n^3)$, following Theorem~\ref{thm:mincutTreeExistence}.
  Then, by running the two algorithms in Lemma~\ref{lem:incomparable} and
  Lemma~\ref{lem:chains}, we compute a solution with $s \le k-1$
  vertices with $\saved$ value at least
  \begin{align*}
    & \max\left\{
      \frac14\e_3(k-1)\minkut - \delta(1+\e_1)(k-1)\minkut,\
      \frac12\ell'(r) - \delta(1+\e_1)(k-1)\minkut \right\} \\
  = & \max\left\{ \frac14\e_3(k-1)\minkut, \ \frac12\ell'(r) \right\} - \delta(1+\e_1)(k-1)\minkut
  \end{align*}
  for each root $r \in V_{\mT}$ (see
  Algorithm~\ref{alg:laminarRooted}). Using $\max\{p, q\} \geq
  (4p+2q)/6$ and 
  $\ell'(r) \ge \ell^*(r) - \e_3(k-1)\minkut$
  we get a solution with $\saved$ value at least
  \begin{align*}
    & \frac16 \left( 4 \cdot \frac14\e_3(k-1)\minkut + 2 \cdot \frac12 \left[ \ell^*(r) - \e_3(k-1)\minkut \right] \right) - \delta(1+\e_1)(k-1)\minkut \\
    \ge & \frac16\ell^*(r) - 2\delta(k-1)\minkut,
  \end{align*}
  using that $\eps_1 \leq 1$. In particular, the best solution
  $v_1,\ldots,v_s \in V_{\mT}$ over all $r$ satisfies
  \[ \saved(v_1,\ldots,v_s) \ge \frac16 \ell^* - 2\delta(k-1)\minkut ,\]
  where $\ell^*(r)$ was replaced by $\ell^*$.

  Let $v_1,\ldots,v_s \in V_{\mT}$ be our solution with
  $\saved(v_1,\ldots,v_s) \ge \frac16 \ell^* - 2\delta(k-1)\minkut$. Let
  $S_1,\ldots,S_s \subset V$ be the corresponding subsets in $V$, i.e.,
  $S_i := \phi^{-1}(\subtree(v_i))$.  Then, add the complement set
  $S_{s+1}:=V \setminus \bigcup_{i\in[s]}S_i$ to the solution, so that
  the sets $S_i$ partition $V$, and
        \[w(E(S_1,\ldots,S_{s+1})) \le s(1+\e_1)\minkut - \left(\frac16 \ell^* -
    2\delta(k-1)\minkut\right). \] Then, extend the solution to a
        $k$-partition using Algorithm~\ref{alg:complete}. We now claim that every additional
        cut that Algorithm~\ref{alg:complete} makes is a $(1+\e_1)$-mincut.
        To see this, observe that $S_1^*, \ldots, S_{k-1}^*$ are all $(1+\e_1)$-mincuts and one
        of them, say $S_j^*$, has to intersect some $S_i$.  Then, the cut $(S_i \cap S_j^*, S_i \setminus S_j^*)$
        is a $(1+\e_1)$-mincut in $S_i$. We can repeat this argument as long as we have $< k$ components $S_i$.
        
        At the end, we have a
  solution $S_1', \ldots, S_k'$ satisfying
  \begin{align*}
          w(E(S_1', \ldots, S_k')) & \le w(E(S_1,\ldots,S_s)) +
    (k-1-s)(1+\e_1)\minkut
    \\
    & \le (k-1)(1+\e_1)\minkut- \left(\frac16 \ell^*  -
      2\delta(k-1)\minkut\right)
  \end{align*}
  Let $S^*_1,\ldots,S^*_k$ be the optimal partition in $\solns_{\e_1}$
  satisfying $\phi(r) \in S^*_k$, and let $\ell^*$ be the maximum of
  $\saved(v_1^*,\ldots,v_{k-1}^*)$ over incomparable
  $v_1^*,\ldots,v_{k-1}^*$. Our solution has approximation ratio
  \begin{align*}
          \frac{w(E(S_1,\ldots,S_k))}{w(E(S_1^*,\ldots,S_k^*))}  & \le
    \frac{(k-1)(1+\e_1)\minkut - \frac16\ell^* +
      2\delta(k-1)\minkut}{(k-1)\minkut - \ell^*} \\ 
    & = \frac{(k-1)(1+\e_1)\minkut - \frac16\ell^* }{(k-1)\minkut -
      \ell^*} + \frac{ 2\delta(k-1)\minkut}{(k-1)\minkut - \ell^*} \\ 
    & \le 2(1+\e_1) - \frac16 + 4\delta,
  \end{align*}
  with the worst case achieved at $\ell^*=\frac12(k-1)\minkut$, which is
  the highest $\ell^*$ can be. Setting $\e_2:=1/6 - 2\e_1-4\delta$
  concludes the proof.
  
  As for running time, we run the algorithms in Lemma~\ref{lem:incomparable} and Lemma~\ref{lem:chains} sequentially, and the final running time is $\tilde 2^{O(k)}f(k)(\tilde O(n^4) + g(n))$. (The $\tilde O(n^4)$ comes from the case when $k\le 4$, in which we solve the problem exactly in $\tilde O(n^4)$ time.)
\end{proof}


\newcommand{\Wext}{\mathsf{Wdeg}}

\section{An FPT-AS for \pvclong}
\label{sec:partial-vc}

Recall the \pvclong (\pvc) problem: the input is a graph $G = (V,E)$
with edge and vertex weights, and an integer $k$. For a set $S$, define
$E_S$ to be the set of edges with at least one endpoint in $S$. The goal
of the problem is to find a set $S$ with size $|S| = k$, minimizing the
weight $w(E_S) + w(S)$, i.e., the weight of all edges hitting $S$ plus
the weight of all vertices in $S$. Our main theorem is the following.
\begin{theorem}[\pvclong]
  \label{thm:pvc}
  There is a randomized algorithm for \pvc on weighted graphs that, for
  any $\delta \in (0,1)$, runs in  $O(2^{k^6/\delta^3} (m +
  k^8/\delta^3)\,n \log n)$ time and outputs a
  $(1+\delta)$-approximation to \pvc with probability $1 - 1/\poly(n)$.
\end{theorem}

We first extend a result of Marx~\cite{Marx07} to give a
$(1+\delta)$-approximation algorithm for the case where $G$ has edge
weights being integers in $\{1, \ldots, M\}$ and no vertex weights, and
then show how to reduce the general case to this special case, losing
only another $(1+\delta)$-factor.

\subsection{Graphs with Bounded Weights}

\begin{lemma}
  \label{lem:pvc-simple}
  Let $\delta \leq 1$. There is a randomized algorithm for the \pvc
  problem on simple graphs with edge weights in $\{1, \ldots, M\}$ (and
  no vertex weights) that runs in $O(m+Mk^4/\delta)$ time, and outputs
a  $(1+\delta)$-approximation with probability at least
  $2^{-(Mk^2/\delta)}$.
\end{lemma}

\begin{proof}
  This is a simple extension of a result for the
  maximization case given by Marx~\cite{Marx07}. We give two algorithms: one for the case when the
  optimal value is smaller than $\tau := Mk^2/\delta$ (which returns the
  correct solution in time, but with probability $2^{-(Mk^2/\delta)}$),
  and another for the case of the optimal value being at least $\tau$
  (which deterministically returns a $(1+\delta)$-approximation in
  linear time). We run both and return the better of the two solutions.

  First, the case when the optimal value is at least $\tau$. Let the
  \emph{weighted degree} of a node $v$, denoted $w(\partial v)$ be
  defined as $\sum_{e: v \in e} w(e)$. Observe that for
  any set $S$ with $|S| \leq k$,
  \[ 0 \leq \sum_{v \in S} w(\partial v) - w(E_S) \leq M\cdot
  \binom{k}{2}. \] Hence, if $S^*$ is the optimal solution and
  $w(E_{S^*}) \geq \tau$, then picking the set of $k$ vertices with the
  least weighted degrees is a $(1+\delta)$-approximation.

  Now for the case when the optimal value is at most $\tau$. In this case,
  the optimal set $S^*$ can have at most $\tau$ edges incident to it, since
  each edge must have weight at least $1$.  Consider the color-coding
  scheme where we independently and uniformly colors the vertices of $G$
  with two colors (red and blue). With probability $2^{-(\tau+k)}$, all the
  vertices in $S^*$ are colored red, and all the vertices in $N(S^*)
  \setminus S^*$ are colored blue. Consider the ``red components'' in
  the graph obtained by deleting the blue vertices. Then $S^*$ is the
  union of one or more of these red components. To find it, define the
  ``size'' of a red component $C$ as the number of vertices in it, and
  the ``cost'' as the total weight of edges in $G$ that are incident to
  it (i.e., cost $= \sum_{e \in E: e \cap C \neq \emptyset} w(e)$.) 

  Now we can use dynamic programming to find a collection of red
  components with total size equal to $k$ and minimum total cost: this
  gives us $S$ (or some other solution of equal cost). Indeed, if we
  define the ``type'' of each component to be the tuple $(s,c)$ where $s
  \in [1\ldots k]$ is the size (we can drop components of size greater
  than $k$) and $c \in [1\ldots \tau]$ is the cost (we can drop all
  components of greater cost). Let $T(s,c)$ be the number of copies of
  type $(s,c)$, capped at $k$. Assume the types are numbered $\tau_1,
  \tau_2, \ldots, \tau_{k\tau}$. Now if $C(i,j)$ is the minimum cost we can
  have with components of type $\leq \tau_i = (s,c)$ whose total size is
  $j$, then
  \[ C(i,j) = \min_{0 \leq \ell \leq T(s,c)} C(i-1, j - \ell s) + \ell
  c. \] Finally, we return the component achieving $C(k\tau,k)$. This can
  all be done in $O(m + k^2\tau)$ time.
\end{proof}

Repeating the algorithm $O(2^{\tau + k} \log n) = O(2^{Mk^2/\delta + k} \log n)$ times and outputting
the best set found in these repetitions gives an algorithm that finds a
$(1+\delta)$-approximation with probability $1 - 1/\poly(n)$.

\subsection{Solving The General Case}

We now reduce the general \pvc problem, where we have no bounds on the
edge weights (and we have vertex weights), to the special case from the
previous section. 

The idea is simple: given a graph $G = (V,E)$ with edge and vertex
weights, we construct a collection of $|V|$ simple graphs $\{ H_v \}_{v
  \in V}$, each defined on the vertex set $V$ plus a couple new nodes,
and having $O(|V|+|E|)$ edges, with each edge-weight $w'(e)$ being an
integer in $\{1, \ldots, M\}$ and $M = O(k/\delta)^2$, and with no vertex
weights. We find a $(1+\delta/2)$-approximate \pvc solution on each $H_v$,
and then output the set $S$ which has the smallest weight (in $G$) among
these. We show how to ensure that $S \sse V$ and that it is a
$(1+\delta)$-approximation of the optimal solution in $G$.

\begin{proof}[Proof of Theorem~\ref{thm:pvc}]
  Let $S^*$ be an optimal solution on $G$. Define the \emph{extended
    weighted degree} of a vertex $v$, denoted by $\Wext(v)$, to be its
  vertex weight plus the weight of all edges adjacent to it.  I.e.,
  $\Wext(v) := w(v) + w(\partial v)$.
        
  Firstly, assume we know a vertex $v^* \in S^*$ with the largest
  $\Wext(v^*)$; we just enumerate over all vertices to find this vertex.
  We now proceed to construct the graph $H_{v^*}$.  Let $L =
  \Wext(v^*)$, and delete all vertices $u$ with $\Wext(u) > L$. Note
  that (a)~any solution containing $v^*$ has total weight at least $L$,
  and (b)~each remaining edge and vertex has weight $\leq L$.

  Assume that $G$ is simple, since we can combine parallel edges
  together by summing their weights.  Create two new vertices $p, q$,
  and add an edge of weight $L k^2$ between them; this ensures that
  neither of these vertices is ever chosen in any near-optimal
  solution.

  Let $\delta' > 0$ be a parameter to be fixed later; think of $\delta'
  \approx \delta$. For each edge $e = (u,v)$ in the edge set $E$ that
  has weight $w(e) < L\delta'/k^2$, remove this edge and add its weight
  $w(e)$ to the weight of both its endpoints $u,v$. Finally, when there
  are no more edges with $w(e) < L\delta'/k^2$, for each vertex $u$ in
  $V$, create a new edge $\{u,p\}$ with weight being equal to the
  current vertex weight $w(u)$, and zero out the vertex weight.  Let the
  new edge set be denoted by $E'$. We claim that for any set $S
  \subseteq V$ of size $\le k$, 
  \[ \left( \sum_{e \in E': e \cap S} w(e)\right) - \left( \sum_{e \in
      E: e \cap S} w(e) + \sum_{v \in S} w(v) \right) \leq \delta' L. \]
  Indeed, the only change comes because of edges with weight $w(e) <
  L\delta'/k^2$ and with both endpoints within $S$---these edges
  contributed once earlier, but replacing them by the two edges means we
  now count them twice. Since there are at most $\binom{k}{2}$ such
  edges, they can add at most $\delta' L$.

  At this point, all edges in the original edge set $E$ have weights in
  $[L \delta'/k^2, Lk^2]$; the only edges potentially having weights $<
  L \delta'/k^2$ are those between vertices and the new vertex $p$. For
  any such edge with weight $< L \delta'/k$, we delete the edge. This
  again changes the optimal solution by at most an additive $L \delta'$,
  and ensure all edges in the new graph have weights in $[L \delta'/k^2,
  Lk^2]$.  Note that since the optimal solution has value at least $L$
  by our guess, these additive changes of $L \delta'$ to the optimal
  solution mean a multiplicative change of only $(1+\delta')$.

  Finally, discretize the edge weights by rounding each edge weight to
  the closest integer multiple of $L \delta'^2/k^2$. Since each edge
  weight $\geq L \delta'/k^2$, each edge weight incurs a further
  multiplicative error at most $1+\delta'$. Note that $M =
  k^4/\delta'^2$. Now use Lemma~\ref{lem:pvc-simple} to get a
  $(1+\delta')$-approximation for \pvc on this instance with high
  probability.   Setting $\delta' = O(\delta)$ ensures that this solution
  is within a factor $(1+\delta)$ of that in $G$.
\end{proof}



\section{Conclusion and Open Problems}
\label{sec:conclusion}
Putting the sections together, we conclude with a proof of our main theorem.

\begin{proof}[Proof of Theorem~\ref{thm:kcut-main}]
  Fix some $\delta \in (0,1/24)$. By Theorem~\ref{thm:pvc}, there is a
  $(1+\delta)$-approximation algorithm for $\PartialVC k$ running in
  time $O(2^{k^6/\delta^3} (m + k^8/\delta^3)\,n \log n) = 2^{O(k^6)}
  n^4$ time. Plugging in $f(k) := 2^{O(k^6)}$ and $g(n) := n^4$ into
  Theorem~\ref{thm:laminar}, we get a $(2-\e_2)$-approximation algorithm
  to $\Laminarcut k{\e_1}$ in time $2^{O(k)}f(k)(n^3 + g(n)) =
  2^{O(k^6)} n^4$, for a fixed $\e_1 \in (0,1/6-4\delta)$. Plugging in
  $f(k):=2^{O(k^6)}$ and $g(n):=n^4$ into Theorem~\ref{thm:reduction1}
  gives a $(2-\e_3)$-approximation for $\kcut$ in time $2^{O(k^2 \log k)} \cdot
  f(k) \cdot (n^4 \log^3 n + g(n)) = 2^{O(k^6)} n^4 \log^3 n$.

  Finally, for our approximation factor. Theorem~\ref{thm:laminar} sets
  $\e_2:=1/6 - 2\e_1 - 4\delta$ for any small enough $\delta$. We can
  take $\e_1$ and $\e_2$ to be equal, so that $\e_1 = \e_2 = 1/18 -
  \nicefrac43\cdot\delta$. Finally, setting
  $\e_4=\e_5=\min(\e_1,\e_2)/3$ and $\e_3:=\e_4^2$ in
  Theorem~\ref{thm:reduction1} gives $\e_3 = 1/54^2 - \delta'$ for some
  arbitrarily small $\delta'>0$. In other words, our approximation
  factor is $2 - 1/54^2 + \delta'$, or $1.9997$ for an appropriately
  small $\delta'$.
\end{proof}

Our result
combines ideas from approximation algorithms and FPT algorithms and
shows that considering both settings simultaneously can help bypass
lower bounds in each individual setting, namely the $W[1]$-hardness of
an exact FPT algorithm and the SSE-hardness of a polynomial-time
$(2-\e)$-approximation. While our improvement is quantitatively modest,
we hope it will prove qualitatively significant.  Indeed, we hope these
and other ideas will help resolve whether an $(1+\eps)$-approximation
algorithm exists in FPT time, and to show a matching lower and upper
bound.

\paragraph{Acknowledgments.} We thank Marek Cygan for generously giving
his time to many valuable discussions.

{\small \bibliographystyle{alpha}
\bibliography{refs}}

\newcommand{\etalchar}[1]{$^{#1}$}
\begin{thebibliography}{DECF{\etalchar{+}}03}

\bibitem[BG97]{BG97}
Michel Burlet and Olivier Goldschmidt.
\newblock A new and improved algorithm for the {$3$}-cut problem.
\newblock {\em Oper. Res. Lett.}, 21(5):225--227, 1997.

\bibitem[BSW17]{BuchbinderSW17}
Niv Buchbinder, Roy Schwartz, and Baruch Weizman.
\newblock Simplex transformations and the multiway cut problem.
\newblock In {\em Proceedings of the Twenty-Eighth Annual {ACM-SIAM} Symposium
  on Discrete Algorithms, {SODA} 2017}, pages 2400--2410, 2017.

\bibitem[CCF14]{CaoCF14}
Yixin Cao, Jianer Chen, and Jia{-}Hao Fan.
\newblock An {$O(1.84^k)$} parameterized algorithm for the multiterminal cut
  problem.
\newblock {\em Inf. Process. Lett.}, 114(4):167--173, 2014.

\bibitem[CCH{\etalchar{+}}16]{Chitnis}
Rajesh Chitnis, Marek Cygan, MohammadTaghi Hajiaghayi, Marcin Pilipczuk, and
  Micha{\l} Pilipczuk.
\newblock Designing {FPT} algorithms for cut problems using randomized
  contractions.
\newblock {\em SIAM J. Comput.}, 45(4):1171--1229, 2016.

\bibitem[CFK{\etalchar{+}}15]{FPT-book}
Marek Cygan, Fedor~V. Fomin, {\L}ukasz Kowalik, Daniel Lokshtanov, D\'aniel
  Marx, Marcin Pilipczuk, Micha{\l} Pilipczuk, and Saket Saurabh.
\newblock {\em Parameterized algorithms}.
\newblock Springer, Cham, 2015.

\bibitem[DECF{\etalchar{+}}03]{DEFPR03}
Rodney~G. Downey, Vladimir Estivill-Castro, Michael Fellows, Elena Prieto, and
  Frances~A. Rosamund.
\newblock Cutting up is hard to do: The parameterised complexity of $k$-cut and
  related problems.
\newblock {\em Electronic Notes in Theoretical Computer Science}, 78:209--222,
  2003.

\bibitem[EG77]{EG75}
Jack Edmonds and Rick Giles.
\newblock A min-max relation for submodular functions on graphs.
\newblock pages 185--204. Ann. of Discrete Math., Vol. 1, 1977.

\bibitem[GH94]{GH94}
Olivier Goldschmidt and Dorit~S. Hochbaum.
\newblock A polynomial algorithm for the {$k$}-cut problem for fixed {$k$}.
\newblock {\em Math. Oper. Res.}, 19(1):24--37, 1994.

\bibitem[HO94]{HO92}
Jianxiu Hao and James~B. Orlin.
\newblock A faster algorithm for finding the minimum cut in a directed graph.
\newblock {\em J. Algorithms}, 17(3):424--446, 1994.
\newblock Third Annual ACM-SIAM Symposium on Discrete Algorithms (Orlando, FL,
  1992).

\bibitem[Kap96]{Kapoor96}
Sanjiv Kapoor.
\newblock On minimum {$3$}-cuts and approximating {$k$}-cuts using cut trees.
\newblock In {\em Integer programming and combinatorial optimization
  ({V}ancouver, {BC}, 1996)}, volume 1084 of {\em Lecture Notes in Comput.
  Sci.}, pages 132--146. Springer, Berlin, 1996.

\bibitem[Kar00]{Karger00}
David~R. Karger.
\newblock Minimum cuts in near-linear time.
\newblock {\em J. ACM}, 47(1):46--76, 2000.

\bibitem[KS96]{KS96}
David~R. Karger and Clifford Stein.
\newblock A new approach to the minimum cut problem.
\newblock {\em Journal of the ACM (JACM)}, 43(4):601--640, 1996.

\bibitem[KT11]{KT11}
Ken-ichi Kawarabayashi and Mikkel Thorup.
\newblock The minimum $k$-way cut of bounded size is fixed-parameter tractable.
\newblock In {\em Foundations of Computer Science (FOCS), 2011 IEEE 52nd Annual
  Symposium on}, pages 160--169. IEEE, 2011.

\bibitem[KV12]{KV12}
Bernhard Korte and Jens Vygen.
\newblock {\em Combinatorial optimization}, volume~21 of {\em Algorithms and
  Combinatorics}.
\newblock Springer, Heidelberg, fifth edition, 2012.
\newblock Theory and algorithms.

\bibitem[KYN07]{KYN06}
Yoko Kamidoi, Noriyoshi Yoshida, and Hiroshi Nagamochi.
\newblock A deterministic algorithm for finding all minimum {$k$}-way cuts.
\newblock {\em SIAM J. Comput.}, 36(5):1329--1341, 2006/07.

\bibitem[Lev00]{Levine00}
Matthew~S Levine.
\newblock Fast randomized algorithms for computing minimum $\{$3, 4, 5,
  6$\}$-way cuts.
\newblock In {\em Proceedings of the eleventh annual ACM-SIAM symposium on
  Discrete algorithms}, pages 735--742. Society for Industrial and Applied
  Mathematics, 2000.

\bibitem[Man17]{Manurangsi17}
Pasin Manurangsi.
\newblock {Inapproximability of Maximum Edge Biclique, Maximum Balanced
  Biclique and Minimum $k$-Cut from the Small Set Expansion Hypothesis}.
\newblock In {\em 44th International Colloquium on Automata, Languages, and
  Programming (ICALP 2017)}, volume~80 of {\em Leibniz International
  Proceedings in Informatics (LIPIcs)}, pages 79:1--79:14, 2017.

\bibitem[Mar07]{Marx07}
D{\'a}niel Marx.
\newblock Parameterized complexity and approximation algorithms.
\newblock {\em The Computer Journal}, 51(1):60--78, 2007.

\bibitem[NI92]{NI92}
Hiroshi Nagamochi and Toshihide Ibaraki.
\newblock Computing edge-connectivity in multigraphs and capacitated graphs.
\newblock {\em SIAM J. Discrete Math.}, 5(1):54--66, 1992.

\bibitem[NI00]{NI00}
Hiroshi Nagamochi and Toshihide Ibaraki.
\newblock A fast algorithm for computing minimum 3-way and 4-way cuts.
\newblock {\em Math. Program.}, 88(3, Ser. A):507--520, 2000.

\bibitem[NKI00]{NKI00}
Hiroshi Nagamochi, Shigeki Katayama, and Toshihide Ibaraki.
\newblock A faster algorithm for computing minimum 5-way and 6-way cuts in
  graphs.
\newblock {\em J. Comb. Optim.}, 4(2):151--169, 2000.

\bibitem[NR01]{NR01}
Joseph Naor and Yuval Rabani.
\newblock Tree packing and approximating {$k$}-cuts.
\newblock In {\em Proceedings of the {T}welfth {A}nnual {ACM}-{SIAM}
  {S}ymposium on {D}iscrete {A}lgorithms ({W}ashington, {DC}, 2001)}, pages
  26--27. SIAM, Philadelphia, PA, 2001.

\bibitem[RS08]{RS02}
R.~Ravi and Amitabh Sinha.
\newblock Approximating {$k$}-cuts using network strength as a {L}agrangean
  relaxation.
\newblock {\em European J. Oper. Res.}, 186(1):77--90, 2008.

\bibitem[SV95]{SV95}
Huzur Saran and Vijay~V. Vazirani.
\newblock Finding $k$-cuts within twice the optimal.
\newblock {\em SIAM Journal on Computing}, 24(1):101--108, 1995.

\bibitem[Tho08]{Thorup08}
Mikkel Thorup.
\newblock Minimum $k$-way cuts via deterministic greedy tree packing.
\newblock In {\em Proceedings of the fortieth annual ACM symposium on Theory of
  computing}, pages 159--166. ACM, 2008.

\bibitem[XCY11]{XCY11}
Mingyu Xiao, Leizhen Cai, and Andrew Chi-Chih Yao.
\newblock Tight approximation ratio of a general greedy splitting algorithm for
  the minimum $k$-way cut problem.
\newblock {\em Algorithmica}, 59(4):510--520, 2011.

\bibitem[ZNI01]{ZNI01}
Liang Zhao, Hiroshi Nagamochi, and Toshihide Ibaraki.
\newblock Approximating the minimum {$k$}-way cut in a graph via minimum 3-way
  cuts.
\newblock {\em J. Comb. Optim.}, 5(4):397--410, 2001.

\end{thebibliography}

\appendix

\section{Pseudocode for \Laminarcut{$k$}{$\eps_1$}}
\label{sec:pseudocode-laminar}

\begin{algorithm}
\caption{SubtreePartialVC$(G,\mT,A,s,\delta)$}
\label{alg:subtreePVC}
\begin{algorithmic}

\If {$|A| < s$}
\State \Return None
\EndIf

\For {$a \in A$}
\State $C_a \gets V(a) \cup \displaystyle\bigcup\limits_{a' \in \desc(a)}V(a')$
\EndFor \Comment{\textbf{Assert}: $C_a$ are all disjoint}

\

\State $\mathcal C \gets \{ C_a : a \in A\}$

\State $H \gets \text{Contract}(G, \mathcal C)$ \Comment{For each $C_a \in \mathcal C$, contract all vertices in $C_a$ into a single vertex in $H$}

\

\For {$i \in [k-1]$}
  \State $P_{i} \gets \text{PartialVC}(H,i)$   \Comment{$P_{i} \in V(H)^i$}
  \State $\mS_i \gets \text{Expand}(H,P_i)$ \Comment { \parbox[t]{.5\linewidth}{ Map each $v \in P_i$ to the set of vertices in $V$ which contract to $v$ in $H$, and call the result $\mS_i \in \left(2^V\right)^i$ } }
\EndFor
\State \Return $\{\mS_{i} : i \in [s]\}$

\end{algorithmic}
\end{algorithm}

\begin{algorithm}
\caption{SingleBranch$(G,\mT,B,k,\delta)$}
\label{alg:singleBranch}
\begin{algorithmic}
\For{$a \in B$}
\State $\Record(\text{SubtreePartialVC}(G, \mT, \children \left((\{a\} \cup \anc(a)) \cap B \right), k-1, \delta))$
\EndFor
\State Return the best recorded solution $\{v_1, \ldots, v_{k-1}\} \in V_{\mT}$.
\end{algorithmic}
\end{algorithm}

\begin{algorithm}
\caption{Laminar$(G=(V,E,w),\mT,k,\eps_1,\delta)$}
\label{alg:laminar}
\begin{algorithmic}

\State $\mT=(V_{\mT},E_{\mT},w_{\mT}) \gets \text{MincutTree}(G)$.
\For{$r \in V_{\mT}$}
\State Root $\mT$ at $r$.
\State $\Record(\text{LaminarRooted}(G,\mT,r,k,\e_1,\delta))$
\EndFor
\State Return the best recorded $k$-partition.

\end{algorithmic}
\end{algorithm}

\begin{algorithm}
\caption{LaminarRooted$(G=(V,E,w),\mT,r,k,\delta_1,\delta)$}
\label{alg:laminarRooted}
\begin{algorithmic}
\For {$a \in V(\mT)$}
  \State $\{S_{a,i} : i \in [k-1]\} \gets \text{SubtreePartialVC}(G, \mT, \children(a), k-1, \delta)$ \Comment{$S_{a,i} \in \left(2^V\right)^i$}
\EndFor

\

\State $A \gets \emptyset$ \Comment{$A \subset V(\mT) \times [k]$ is the set of \textit{anchors}}
\For {$a \in V(\mT)$ in topological order from leaf to root}
  \State $\e_3 \gets \frac{1-\delta}4-2\e_1$ \Comment {The optimal value of $\e_3$}
  \State $I_a \gets \{i \in [k-1] : \text{Value}(P_{a,i}) \ge \e_3(1-\delta)(i-1)\minkut \}$
  \If {$I_a \ne \emptyset$ \textbf{and} $\nexists (a',i) \in A : a' \in \desc(a)$} \Comment{Only take \textit{minimal} anchors}
    \State $A \gets A \cup \{(a, \max I_a)\}$
  \EndIf
\EndFor

\

\If {$|A| \ge k-1$} \Comment{\textbf{Case (K)}: Knapsack}
  \State $A' \gets \text{Knapsack}(A)$ \Comment{The Knapsack algorithm as described in Lemma~\ref{lem:incomparable}}
  \State $\mS \gets \displaystyle\bigcup\limits_{(a,i) \in A} \{S_{a,i}\}$  \Comment{The partition for Case (K), to be computed. \textbf{Assert}: $|S| \le k-1$}
  \State $\Record(\text{Complete}(G,k,\mS))$
\Else
  \State $\mB \gets \text{Branches}(A)$ \Comment{$\mB \subset \left( 2^{V(\mT)} \right)^r$ for some $k-1 \le r \le 2k-3$}

  \

  \For {$B \in \mB$}  \Comment {\textbf{Case (B1)}: Compute branches independently}
    \State $\{P_{B,i} : i \in [k-1]\} \gets \text{SingleBranch}(G, \mT, B, k-1, \delta)$ \Comment{$P_{B,i} \in V^i$}
  \EndFor
        \State $(\mB^*,\mathbf i^*) \gets \argmin\limits_{ \substack {\mB' \subset \mathcal B \text{ incomparable},\\ \mathbf i \in [k-1]^{\mB'} : \ \sum_{B} i_B\ =\ k-1} } \ \displaystyle\sum\limits_{B \in \mB'} w(E(P_{B,i_B}))$ \Comment{Computed by brute force}
  \State $\mS_1 \gets \bigcup_{B \in \mB^*} \{P_{B, i_B}\}$ \Comment{The partition in Case (B1)}
  \State $\Record(\text{Complete}(G,k,\mS_1))$

  \

  \For {$B \in \mB$} \Comment{\textbf{Case (B2)}: Guess the branches with the anchors}
    \State $a_B \gets (a \in B : B \setminus a \subset \desc(a))$ \Comment{$a_B$ is the common ancestor of branch $B$}
  \EndFor
  \For {$\mB' \subset \mB$ s.t.\ $\nexists B_1,B_2\in\mB' : B_1 \subset \desc(B_2)$} \Comment{Subsets whose branches are incomparable}
    \State $A_{\mB'} \gets \children\left(\bigcup_{B \in \mB'} \left( \{a_B\} \cup \anc(a_B) \right) \right)$
    \State $\mS_{2,\mB'} \gets \text{SubtreePartialVC}(G,\mT,A_{\mB'},k-1,\delta)$ \Comment{The partition for $\mB'$ in Case (B2)}
    \State $\Record(\text{Complete}(G,k,\mS_{2,\mB'}))$
  \EndFor
\EndIf

\

\State Return the best recorded $k$-partition.

\end{algorithmic}
\end{algorithm}


\newpage

\section{Missing Proofs}
\label{sec:missing-proofs}

\begin{lemma} \label{lemma:knapsack}
Consider the knapsack instance of $k-1$ items $i \in [k-1]$ where item $i$ has size $s_i \in [2,k-1]$ and value $s_i-1$. There is an algorithm achieving value $\ge (k-1)/4$ for $k \ge 5$, running in $O(k)$ time.
\end{lemma}

\begin{proof}
Consider the greedy
    knapsack solution where we always choose the heaviest item, if still
    possible. Let $A \in [k-1]$ be our solution. If our total size
    $\sum_{i\in A}s_i$ is at least $k - 1 - \sqrt k$, then our value is
    at least $\sum_{i\in A}(s_i-1) \ge \sum_{i\in A}s_i/2 \ge (k-1-\sqrt
    k)/2$.  Otherwise, since we could not fit the next item of size at
    least $\sqrt k$ into our solution, all of our items have size at
    least $\sqrt k$. Furthermore, our total solution size is at least
    $(k-1)/2$, so $\sum_{i \in A}(s_i-1) \ge \sum_{i\in A}(1-1/\sqrt
    k)s_i \ge (1-1/\sqrt k)(k-1)/2$. When $k \ge 5$, the
    value is $\ge (1-1/\sqrt 5)(k-1)/2 \ge (k-1)/4$.
\end{proof}


\end{document}